\documentclass[aps,prd,amssymb,superscriptaddress,nofootinbib,11pt,notitlepage,twoside]{revtex4-1}
 \usepackage{setspace} 
\usepackage{enumerate}
\usepackage{graphicx,epic,eepic,epsfig,amsmath,latexsym,amssymb,verbatim,revsymb,color,stmaryrd,fancyhdr,hyperref}

\newcommand{\binomial}[2]{\genfrac{(}{)}{0pt}{}{#1}{#2}}

\usepackage{theorem}
\newtheorem{definition}{Definition}

\newtheorem{lemma}[definition]{Lemma}

\newtheorem{theorem}[definition]{Theorem}
\newtheorem{corollary}[definition]{Corollary}

\def\squareforqed{\hbox{\rlap{$\sqcap$}$\sqcup$}}
\def\qed{\ifmmode\squareforqed\else{\unskip\nobreak\hfil
\penalty50\hskip1em\null\nobreak\hfil\squareforqed
\parfillskip=0pt\finalhyphendemerits=0\endgraf}\fi}
\def\endenv{\ifmmode\;\else{\unskip\nobreak\hfil
\penalty50\hskip1em\null\nobreak\hfil\;
\parfillskip=0pt\finalhyphendemerits=0\endgraf}\fi}

\newlength{\blank}
\newenvironment{proof}[1][{\hspace{-\blank}}]{{\noindent\textbf{Proof   }}}{\hfill\qed\vskip 0.5\baselineskip}

\mathchardef\ordinarycolon\mathcode`\:
\mathcode`\:=\string"8000
\def\vcentcolon{\mathrel{\mathop\ordinarycolon}}
\begingroup \catcode`\:=\active
  \lowercase{\endgroup
  \let :\vcentcolon
  }

\newcommand{\nc}{\newcommand}
\nc{\rnc}{\renewcommand}

\nc{\lbar}[1]{\overline{#1}}
\nc{\bra}[1]{\langle#1|}
\nc{\ket}[1]{|#1\rangle}
\nc{\ketbra}[2]{|#1\rangle\!\langle#2|}
\nc{\braket}[2]{\langle#1|#2\rangle}
\nc{\proj}[1]{| #1\rangle\!\langle #1 |}
\nc{\avg}[1]{\langle#1\rangle}
\nc{\Rank}{\operatorname{rank}\,}
\nc{\smfrac}[2]{\mbox{$\frac{#1}{#2}$}}
\nc{\tr}{\operatorname{Tr}}

\nc{\cA}{{\cal A}}
\nc{\cB}{{\cal B}}
\nc{\cC}{{\cal C}}
\nc{\cD}{{\cal D}}
\nc{\cE}{{\cal E}}
\nc{\cF}{{\cal F}}
\nc{\cG}{{\cal G}}
\nc{\cH}{{\cal H}}
\nc{\cI}{{\cal I}}
\nc{\cJ}{{\cal J}}
\nc{\cK}{{\cal K}}
\nc{\cL}{{\cal L}}
\nc{\cM}{{\cal M}}
\nc{\cN}{{\cal N}}
\nc{\cO}{{\cal O}}
\nc{\cP}{{\cal P}}
\nc{\cR}{{\cal R}}
\nc{\cS}{{\cal S}}
\nc{\cT}{{\cal T}}
\nc{\cU}{{\cal U}}
\nc{\cX}{{\cal X}}
\nc{\cZ}{{\cal Z}}

\nc{\1}{{\openone}}

\def\a{\alpha}

\def\g{\gamma}

\def\m{\mu}
\def\n{\nu}

\def\s{\sigma}

\nc{\RR}{{{\mathbb R}}}
\nc{\CC}{{{\mathbb C}}}
\nc{\FF}{{{\mathbb F}}}
\nc{\NN}{{{\mathbb N}}}
\nc{\ZZ}{{{\mathbb Z}}}
\nc{\PP}{{{\mathbb P}}}
\nc{\QQ}{{{\mathbb Q}}}
\nc{\UU}{{{\mathbb U}}}
\nc{\EE}{{{\mathbb E}}}
\nc{\id}{{\operatorname{id}}}

\nc{\be}{\begin{equation}}
\nc{\ee}{\end{equation}}
\nc{\bea}{\begin{eqnarray}}
\nc{\eea}{\end{eqnarray}}
\nc{\<}{\langle}
\rnc{\>}{\rangle}

\nc{\LO}{\text{LO}}
\nc{\LOCC}{\text{LOCC}}
\nc{\cLOCC}{{\overline{\text{LOCC}}}}
\nc{\SEP}{\text{SEP}}
\nc{\PPT}{\text{PPT}}

\nc{\sep}{\text{sep}}
\nc{\twist}{\text{twist}}

\nc{\te}{\otimes}

\newcommand*{\pro}[1]{\ket{#1}\!\bra{#1}}

\begin{document}

\singlespacing
\title{Witnessing entanglement by proxy}

\author{Stefan B{\"a}uml}
\email{stefan.baeuml@bristol.ac.uk}
\affiliation{Department of Mathematics, University of Bristol, Bristol BS8 1TW, UK}
\affiliation{F\'{\i}sica Te\`{o}rica: Informaci\'{o} i Fen\`{o}mens  Qu\`{a}ntics, Universitat Aut\`{o}noma de Barcelona, ES-08193 Bellaterra (Barcelona), Spain}
              
\author{Dagmar Bru\ss}
\email{bruss@thphy.uni-duesseldorf.de}
\affiliation{Institut f{\"u}r Theoretische Physik III, Heinrich-Heine-Universit{\"a}t D{\"u}sseldorf, Universit{\"a}tsstra{\ss}e 1, Geb{\"a}ude 25.32, D-40225 D{\"u}sseldorf, Germany}

\author{Marcus Huber}
\email{marcus.huber@univie.ac.at}
\affiliation{F\'{\i}sica Te\`{o}rica: Informaci\'{o} i Fen\`{o}mens  Qu\`{a}ntics, Universitat Aut\`{o}noma de Barcelona, ES-08193 Bellaterra (Barcelona), Spain}
\affiliation{ICFO-Institut de Ci\`{e}ncies Fot\`{o}niques, Mediterranean Technology Park, 08860 Castelldefels (Barcelona), Spain}

\author{Hermann Kampermann}
\email{kampermann@thphy.uni-duesseldorf.de}
\affiliation{Institut f{\"u}r Theoretische Physik III, Heinrich-Heine-Universit{\"a}t D{\"u}sseldorf, Universit{\"a}tsstra{\ss}e 1, Geb{\"a}ude 25.32, D-40225 D{\"u}sseldorf, Germany}

\author{Andreas Winter}
\email{andreas.winter@uab.cat}
\affiliation{ICREA - Instituci\'{o} Catalana de Recerca i Estudis Avan\c{c}ats, ES-08010 Barcelona, Spain}
\affiliation{F\'{\i}sica Te\`{o}rica: Informaci\'{o} i Fen\`{o}mens Qu\`{a}ntics, Universitat Aut\`{o}noma de Barcelona, ES-08193 Bellaterra (Barcelona), Spain}

\begin{abstract}
Entanglement is a ubiquitous feature of low temperature systems and believed to be highly relevant for the dynamics of condensed matter properties and quantum computation even at higher temperatures. The experimental certification of this paradigmatic quantum effect in macroscopic high temperature systems is constrained by the limited access to the quantum state of the system. In this paper we show how macroscopic observables beyond the energy of the system can be exploited as proxy witnesses for entanglement detection. Using linear and semi-definite relaxations we show that all previous approaches to this problem can be outperformed by our proxies, i.e. entanglement can be certified at higher temperatures without access to any local observable. For an efficient computation of proxy witnesses one can resort to a generalized grand canonical ensemble, enabling entanglement certification even in complex systems with macroscopic particle numbers.
\end{abstract}

\maketitle

\renewcommand{\tocname}{}

\phantom{.}\vspace{-1.99cm}

\tableofcontents

\section{Introduction}
While the occurrence and possible uses of entanglement were first studied for bipartite states, entanglement in systems containing a large number of particles is of interest both from a theoretical and from a practical point of view. Even though the macroscopic world we experience daily, can be described classically, there are a number of systems that are large enough to be described by the thermodynamic limit, which exhibit quantum behaviour, Bose-Einstein condensates, ferromagnetic and superconducting materials being prominent examples. Entanglement may turn out useful in understanding thermodynamic phenomena such as phase transitions in such systems \cite{osterloh2002scaling,osborne2002entanglement,Amico2008}. Recently there has also been a lot of attention on the role of entanglement in quantum thermodynamics \cite{brunner2014entanglement,quantacell,review2015}. Other possible applications of large entangled systems are quantum computers based on solid state or NMR systems \cite{ekert1998quantum,ladd2010quantum,PhysRevLett.86.5188,benjamin2009prospects}. In addition to studying entanglement in the limit of many particles it is also worth asking up to which temperature entanglement can exist. This is an important question for experiments, where cooling down systems requires lots of resources. While entanglement usually exists at very small temperatures, it could persist to up to 100K in superconductors \cite{vedral2004high}.

Experimentally detecting entanglement in macroscopic systems is generally a highly non-trivial task. Checking for instance the famous PPT (positivity under partial transpose) criterion, as easy as it is theoretically, requires a full state tomography, which is not possible in large systems. Also, calculating the eigenvalues for matrices of large dimensions is not practical. The method of choice are entanglement witnesses, i.e. observables with positive expectation value for all separable states but with negative expectation value for some entangled states. Witnesses reduce the complexity of entanglement detection to the measurement of a single observable. However, this observable might have no physical meaning and might be hard or impossible to measure. In particular it might be necessary to perform a collective measurement of all particles, which is not experimentally feasible in macroscopic systems. What is feasible is the measurement of macroscopic observables such as the mean energy, the magnetisation, the temperature or the entropy of the system. There have been several results showing that mean energy and temperature can serve as entanglement witnesses at low temperatures (\cite{dowling2004energy,BV04,toth2005entanglement,PhysRevA.72.032309,anders2006detecting,GH13} to name just a few).
 
However, all of these are limited intrinsically at higher temperatures when
the value of the macroscopic witness is consistent with separable pure states. Knowing that any experimental system will have some non-zero entropy, i.e. not be in pure state, often allows for reasonable lower bounds on the system entropy to be assumed (as for example through ambient temperature and the second law of thermodynamics).
Here we can leverage the entropy to bound ``by proxy''
generic entanglement criteria, allowing us to detect, in principle, 
all entangled Gibbs states of a many-body system, as well as entangled states far out of equilibrium. This method can in general be phrased as a semi-definite program (SDP), which are efficiently solvable for small system sizes (and have in fact often been used in the context of entanglement quantification \cite{ref1,ref2,ref3}). In the following we will showcase some exemplary situations where these SDPs can improve entanglement detection for systems of up to thirteen qubits on a regular laptop. Furthermore we show that using entanglement witnesses in particular allows one to harness tools from statistical physics and solve the problem through introducing a virtual ``chemical potential'' in a generalized Gibbs ensemble, changing the maximum entropy state for a given energy. In these generalized Gibbs ensembles entanglement witnesses play the role of additional conserved quantities, making the proxy method as accessible for large systems as the computation of Gibbs entropies, which we demonstrate by detecting entanglement by proxy in the thermodynamic limit. 

In what follows we will assume the Hamiltonian to be well characterized, which for large systems is of course only an approximation. Unfortunately, to characterize entanglement one requires reasonably precise knowledge of the observables used to certify it. As we will show later, the method is robust against small perturbations, but if the Hamiltonian is entirely different it can of course lead to false positives. Our method is suitably generic and requires knowledge only of conserved quantities for use in the generalized Gibbs ensembles. While we generically use the average energy to showcase our methods, they could just as well be replaced by other macroscopic approaches such as e.g. spin squeezing \cite{Squeezy1,Squeezy2,Squeezy3}.

\section{Entropy as entanglement witness}
\subsection{The primal problem}
As mentioned in the introduction, the mean energy of a system can be used to witness entanglement in the corresponding quantum state. Namely any state with mean energy less than
\be\label{eq1}
E_{\text{min,sep}}=\min_{\rho\in\sep}\tr\rho H
\ee
is entangled. By convexity the minimum is attained in a pure state. $\sep$ can be chosen to be the set of fully separable states or the set of $k$-separable states. For $k=2$, genuinely multipartite entanglement is detected. If the system is in thermal equilibrium, it is possible to derive an analogous criterion for the temperature.

The goal of this chapter is to find a condition that is able to detect entanglement at higher energies than $E_{\text{min,sep}}$. The idea is to add additional constraints to (\ref{eq1}). 
 For example, a lower bound on the (von Neumann) entropy of $\rho$:
\be\label{eq2}
E_{\text{min,sep,S'}}=\min_{\rho\in\sep,\;S(\rho)\ge S'}\tr\rho H.
\ee
If $E(\rho)<E_\text{min,sep,S'}$ and $S(\rho)\ge S'$, $\rho$ is entangled. $S'$ can be varied between $0$ and $\ln d$. Clearly, it holds $E_\text{min,sep,S'}\ge E_\text{min,sep}$. But the relevant question is whether there are entropies for which the strict inequality holds. While (\ref{eq1}) is minimised by a pure state, the minimum does not have to be unique. Hence a mixture of minimisers could also have mean energy $E_{\text{min,sep}}$, but at non-zero entropy. The exact behaviour depends, of course, on the Hamiltonian. However, it is possible to show that if there exists an $S'$, for which the strict inequality holds, it will hold for any larger entropy. This follows from
\begin{lemma}
$E_{\text{min,sep,S'}}$ as function of $S'$ is convex.
\end{lemma}
\begin{proof}
Let $0\le S_{1,2}\le\ln d$ and $0\le p\le 1$. Let $\rho_{1,2}$ be separable states with entropies $S'_{1,2}$ respectively, that minimise (\ref{eq2}). Then
\be
E_{\text{min,sep,S'}}\left(pS'_1+(1-p)S'_2\right)\le\tr H\left(p\rho_1+(1-p)\rho_2\right)=pE_{\text{min,sep,S'}}(S'_1)+(1-p)E_{\text{min,sep,S'}}(S_2),
\ee
where the inequality is due to the fact that $p\rho_1+(1-p)\rho_2$ is a feasible point of (\ref{eq2}), which follows from the concavity of the entropy.
\end{proof}
Since $E_{\text{min,sep,S'}}$ is convex, it is strictly monotonically increasing as soon as it exceeds $E_{\text{min,sep}}$. Let us call the smallest entropy where the constraint hits $S_\text{min}$. See also figure \ref{fig2}.

\begin{figure}
\centering
\includegraphics[width=0.5\textwidth]{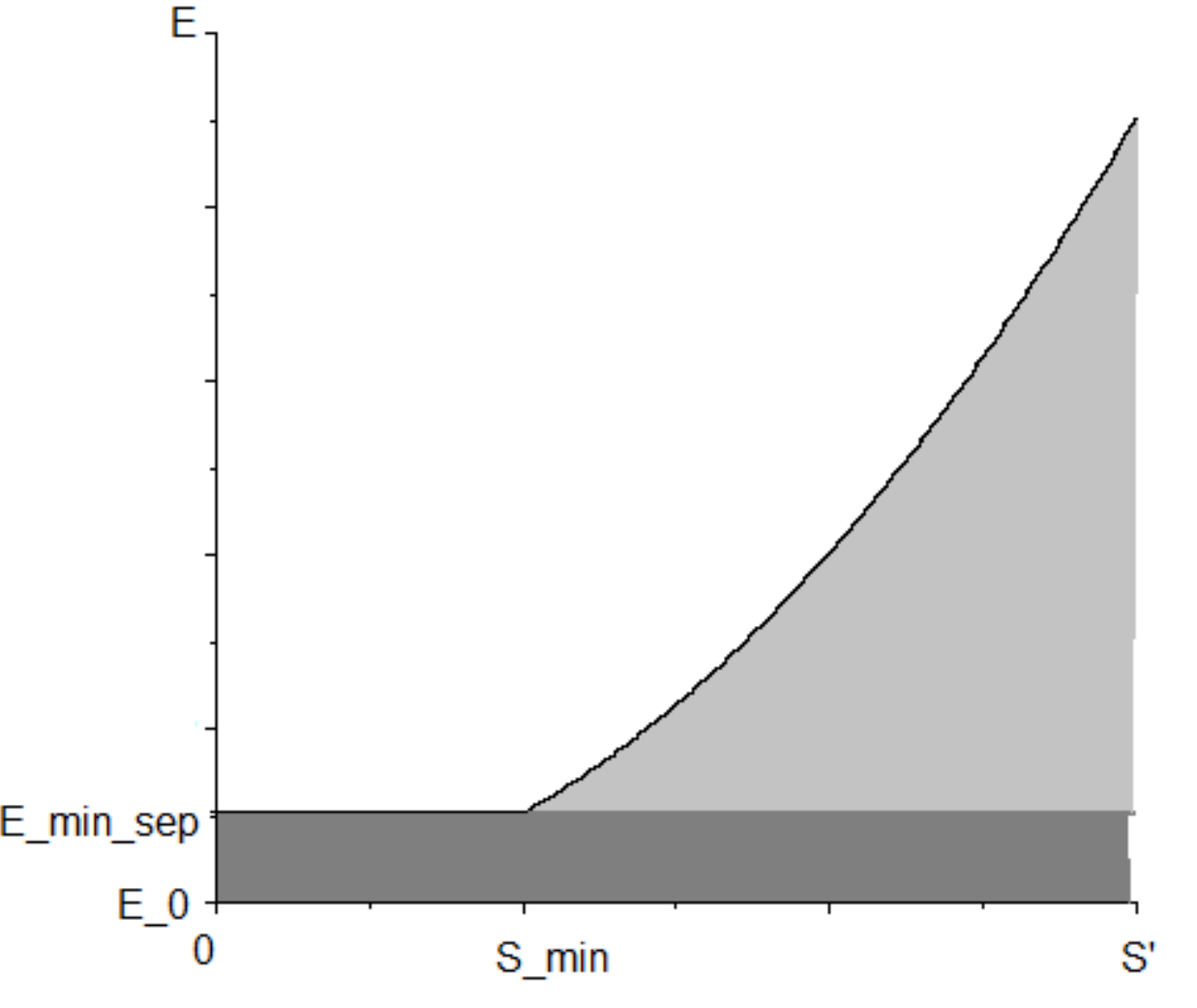}
\caption{\label{fig2} Qualitative form of $E_{\text{min,sep,S'}}$ as a function of $S'$. The dark grey area under $E_\text{min,sep}$ is detected by \cite{dowling2004energy,BV04,toth2005entanglement}. The method presented here can detect entanglement in the light grey area, i.e. for $S\ge S_\text{min}$ and $E<E_{\text{min,sep,S'}}$. In the white area above  the $E_{\text{min,sep,S'}}$ curve, separability is possible. Because of it's  convexity, the same curve can be obtained by maximising the entropy under the constraint that the mean energy is less than $E$ and varying $E$.}
\end{figure}

Since for $S'\ge S_{\text{min}}$, $E_{\text{min,sep,S'}}$ is strictly monotonically increasing, we can obtain the same curve, i.e. detect the same entanglement, by computing
\be\label{eq3}
S_\text{max,sep,E}=\max_{\rho\in\sep,\;\tr H\rho= E}S(\rho)
\ee
and varying $E_{min,sep}\le E\le\frac{\tr H}{d}$. Equivalently we could also demand $\tr H\rho\le E$.  Note that if we remove the separability constraint (\ref{eq3}) will become the Gibbs state entropy
\be\label{eq4}
S_{\text{Gibbs}}=\max_{\rho\text{ state},\;\tr H\rho= E}S(\rho),
\ee
attained by the Gibbs state. Since the optimisation is of the same form as the one yielding the regular Gibbs ensemble we will call the resulting state a \textit{separable Gibbs ensemble}. In order to compare its entropy with the Gibbs state entropy, let us define the entropy gap $
\Delta S=S_\text{Gibbs}-S_\text{max,sep,E}$.
\begin{theorem}\label{th2a}
For mean energy $E$ and corresponding temperature $T$, any state $\rho$ with $S_\text{max,sep,E}<S(\rho)\le S_\text{Gibbs}$ is entangled. In particular the Gibbs state is entangled if $\Delta S>0$.
\end{theorem}
Note that for any non-zero gap between $S_\text{max,sep,E}$ and $S_\text{Gibbs}$ there always exist states different from the Gibbs state which we detect, for example there is always a state with entropy $\lambda S_\text{max,sep,E}+(1-\lambda)S_\text{Gibbs}$, where $0\leq\lambda\leq1$.

The optimisation in (\ref{eq3}) is difficult to deal with because it contains the separability constraint. The main idea is to relax this constraint using sets of states which remain positive semidefinite after the application of positive (yet not completely positive) maps $\Lambda$. A prominent example would be the partial transposition. While these are of course supersets of the separable states, it is clear that for every entangled state in principle there exists a map $\Lambda$ and thus a semidefinite relaxation that will still yield optimal results for the constrained optimisation. 

Let us use the fact that every set of $\Lambda$-positive states forms a convex set that can be approximated by a suitable set of entanglement witnesses $W_i$ with suitable weights $\nu_i\ge0$ \cite{cecilia}. As every optimal entanglement witness for the set of $\Lambda$-positive states for partition $A$ can be written as $\Lambda_A^*\te\1_{\overline{A}}[|\psi\rangle\langle\psi|]$ one can find suitable entanglement witnesses for the system in question. While this of course comes at the expense of finding suitable entanglement witnesses for specific systems, it also opens the possibility to constrain the entropy beyond just states which are separable under fixed bi-partitions. I.e. it enables us to find also genuine multipartite entanglement or any other non-partially separable set, by choosing corresponding witnesses.
In order to remain fully general we include both an arbitrary set of entanglement witnesses and positive maps in the following considerations.

The relaxed problem is thus given by
\begin{align}\label{eq:constrained-entropy}
S_\text{max,sep,E}\leq S_{\text{max},\Lambda_A,W_i,E}=&\max S(\rho)\\
& \text{s.t. }\rho\text{ state},\tr\left(\rho H\right) = E,\nonumber\\
&\Lambda_A\otimes\1_{\overline{A}}[\rho]\geq 0,\tr(\rho W_i)\ge0.\nonumber  
\end{align}
$S_{\text{max},\Lambda_A,W_i,E}$ will be an upper bound on $S_\text{max,sep,E}$, it's tightness depending on the choice of witnesses or maps. Since the relaxation just provides an upper bound for the separable entropy it trivially follows that according to Theorem \ref{th2a}:
\begin{corollary}\label{th2}
For mean energy $E$ and corresponding temperature $T$, any state $\rho$ with $S_{\text{max},\Lambda_A,W_i,E}<S(\rho)\le S_\text{Gibbs}$ is entangled. In particular the Gibbs state is entangled if $S_{\text{max},\Lambda_A,W_i,E}< S_\text{Gibbs}$.
\end{corollary}
While the relaxation weakens the detection criterion, we are now dealing with a convex optimisation problem with only linear and semidefinite constraints (since maximising a concave function is of course equivalent to minimising a convex one). Such programs can be solved efficiently numerically thanks to so-called \textit{interior point methods}\cite{B04}. In addition, they have a duality theory which can be used to give certified upper bounds on $S_{\text{max},\Lambda_A,W_i,E}$, as we will discuss in the next section.
\begin{figure}
\centering
\includegraphics[width=0.7\textwidth]{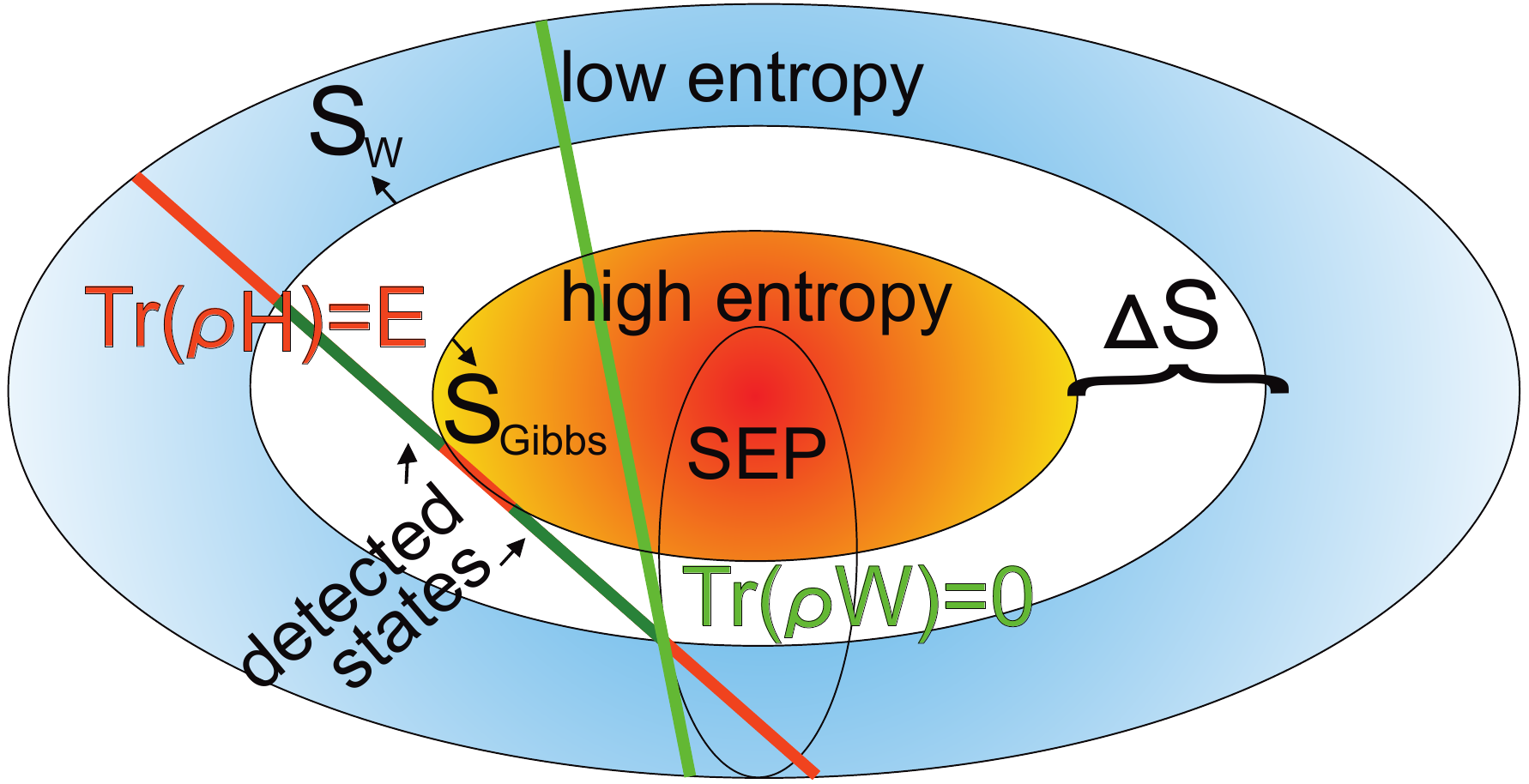}
\caption{\label{fig:entropy}Schematic view of Corollary \ref{th2} for one witness $W$. The constrained entropy $S_{\text{max},W,E}$ or short $S_{W}$ is attained at the intersection between the (green) hyperplane induced by the witness $\text{Tr}(\rho W)=0$ and the (red) hyperplane induced by the mean energy $\text{Tr}(\rho H)=E$. The concentric circles represent equi-entropic states, with the outside low entropy region (blue) bounded by $S_{W}$ and the center high entropy region (red) bounded by the Gibbs state entropy $S_{Gibbs}$. The detected states lie within the (white) region spanned by $\Delta S$.}
\end{figure}
One might ask why instead of applying Corollary \ref{th2}, one cannot simply apply the witnesses or positive maps involved directly. While this is possible theoretically, let us note again that witnesses or positive maps are in general not easily accessible in experiments, while macroscopic variables such as the mean energy and entropy are. Since Corollary \ref{th2} makes use of witnesses without the need to measuring them directly, but instead requires the measurement of mean energy and entropy, we call those two quantities \textit{Proxy Witnesses}.

Let us conclude this section by noting that apart from the von Neumann entropy, in principle any concave function can serve as a proxy witness. An example is the so called \textit{linear entropy}, also known as \textit{impurity} of a quantum state, which is defined by $S_L(\rho)=1-\tr\rho^2$ \cite{PhysRevA.70.052309}. The linear entropy is upper bounded by the von Neumann entropy and can be seen as a measure of mixedness of a quantum state, as well. The optimisation of the linear entropy is easier to deal with than of the von Neumann entropy, as it is only quadratic in $\rho$.

\subsection{The dual problem.}
For convex optimisation problems, such as \ref{eq:constrained-entropy}, it is possible to derive a \textit{dual} optimisation problem. 
To do so, one defines the \textit{Lagrangian}, a function incorporating both objective, i.e. the function to be maximised, and the constraints. The constraints are added by means of Lagrange multipliers. The Lagrange multipliers are referred to as \textit{dual variables}, whereas the variables of the original problem are referred to as \textit{primal}. Maximisation of the Lagrangian over all admissible primal points yields the \textit{dual objective}. The dual problem is then given by the minimisation of the dual objective with respect to the dual variables. It can be shown that any dual feasible point provides an upper bound on the original problem, which is referred to as \textit{weak duality}. For details, please refer to \cite{B04}. For \ref{eq:constrained-entropy},the Lagrangian is given by
\begin{equation}\begin{split}
  L(\rho,\lambda,\mu,\nu_i,X_0,X_{A,\Lambda})
      = S(\rho) &+ \lambda(\tr\rho-1) + \mu(\tr\rho H - E) + \sum_i \nu_i\tr\rho W_i\\
                &+ \tr\rho X_0 + \sum_{A,\Lambda} \tr\Lambda_A\otimes\1_{\overline{A}}[\rho]X_{A,\Lambda},
  \label{eq:Lagrangian}
\end{split}\end{equation}
where $\lambda,\mu$ and the $\nu_i$ are real Lagrange dual variables, corresponding to the trace, mean energy and witness constraints, respectively. As the witness constraint is given by an inequality, we can restrict to $\nu_i\ge0$. $X_0$ and the $X_{A,\Lambda}$, the Lagrange dual variables corresponding to the positivity and positive maps constraints respectively, are positive semidefinite matrices.

Here it becomes clear that witnesses are easier to deal with numerically than positive maps, as they only require scalar variables. The \emph{dual objective function} reads as follows
\begin{equation}
  \ell(\lambda,\mu,\nu_i,X_0,X_{A,\Lambda}) := \max_\rho L(\rho,\lambda,\mu,\nu_i,X_0,X_{A,\Lambda}),
  \label{eq:dual-fct}
\end{equation}
where the maximisation is over all admissible $\rho$, i.e. those for which the right hand side is defined, in particular not necessarily obeying the primal constraints in (\ref{eq:constrained-entropy}). Note that the dual problem is also a convex optimisation problem: the
constraints are indeed linear and semidefinite, while $\ell$ is convex, owed to the linearity of $L$ in the dual variables.

Clearly, for every primal feasible $\rho$ and dual feasible $\lambda,\mu,\nu_i,X_0$ and $X_{A,\Lambda}$ it holds
\be
  S(\rho) \leq L(\rho,\lambda,\mu,\nu_i,X_0,X_{A,\Lambda}) \leq \ell(\lambda,\mu,\nu_i,X_0,X_{A,\Lambda})
\ee
In particular we have what is referred to as \emph{weak duality}:
\begin{equation}
  S_{\text{max},\Lambda_A,W_i,E} \leq \min \ell(\lambda,\mu,\nu_i,X_0,X_{A,\Lambda}),
  \label{eq:weak-duality}
\end{equation}
where the minimisation is over the dual feasible region, i.e. $\lambda,\mu\in\RR,\nu_i\ge0,X_0\ge0,X_{A,\Lambda}\ge0$. The right hand side is referred to as the \textit{Lagrange dual problem}. Note that any dual feasible bound yields an upper bound on $S_{\text{max},\Lambda_A,W_i,E}$. This allows us to obtain analytical upper bounds from numerics. One simply has to numerically optimise the dual problem and check if the so obtained optimisers are dual feasible. If they are they can be inserted into $\ell$ yielding the bound.

Let us now evaluate the dual objective function (\ref{eq:dual-fct}) for the von Neumann entropy, which is well-defined on any positive semidefinite matrix. Since the entropy is concave, so is the Lagrangian. Hence it is sufficient to find a critical point, that is positive semidefinite, of $L$ as function of $\rho$. Using $\nabla_\rho S(\rho) = -\1 - \ln\rho$ and $\nabla_\rho \tr\rho M = M$ with respect to the trace inner product on Hermitian matrices, as well as the fact that $\tr\left( \Lambda_A\otimes\1_{\overline{A}}[\rho]X_{A,\Lambda}\right)=\tr\left( \rho\Lambda^*_A\otimes\1_{\overline{A}}[X_{A,\Lambda}]\right)$ one can obtain
\be
  \nabla_\rho L=-\1-\ln\rho + \lambda\1 + \mu H + \sum_i \nu_i W_i+ X_0 + \sum_{A,\Lambda} \Lambda^*_A\otimes\1_{\overline{A}}[X_{A,\Lambda}],
\ee
which vanishes for
\be
  \rho^\text{crit} = \exp\left( (\lambda-1)\1 + \mu H + \sum_i \nu_i W_i + X_0 + \sum_{A,\Lambda} \Lambda^*_A\otimes\1_{\overline{A}}[X_{A,\Lambda}] \right).
\ee
Hence the dual objective is given by
\begin{align}
\ell(\lambda,\mu,\nu_i,X_0,X_{A,\Lambda})&=L(\rho^\text{crit},\lambda,\mu,\nu_i,X_0,X_{A,\Lambda})\\
&=\tr\exp\left((\lambda-1)\1 + \mu H + \sum_i \nu_i W_i + X_0 + \sum_{A,\Lambda} \Lambda^*_A\otimes\1_{\overline{A}}[X_{A,\Lambda}] \right)-\lambda-\mu E\\
&=e^{(\lambda-1)}\tr\exp\left(\mu H + \sum_i \nu_i W_i + X_0 + \sum_{A,\Lambda} \Lambda^*_A\otimes\1_{\overline{A}}[X_{A,\Lambda}] \right)-\lambda-\mu E.\label{eq:vN-ent-dual-fct}
\end{align}
We will have to minimise $\ell$ with respect the dual variables. Since $X_0\ge0$, the minimum will be attained at $X_0=0$. $\ell$ can be easily minimised for $\lambda$. This is due to the fact that the eigenvalues of a hermitian matrix cannot decrease if a positive semi-definite matrix is added, which follows from Theorem 4.3.1 of \cite{HJ90}. By convexity the minimum is attained where the derivatives vanishes, yielding a function $\tilde{\ell}$ just of $\mu,\nu_i,X_{A,\Lambda}$.
\begin{equation}\label{ell}
	\tilde{\ell}(\mu,\nu_i,X_{A,\Lambda})=\ln\tr\exp\left(\mu H + \sum_i \nu_i W_i + \sum_{A,\Lambda} \Lambda^*_A\otimes\1_{\overline{A}}[X_{A,\Lambda}]\right)-\m E.
\end{equation}
By weak duality, it holds
\begin{equation}
  S_{\text{max},\Lambda_A,W_i,E} \leq \min_{\mu\in\RR,\nu_i\ge0,X_{A,\Lambda}\ge0} \tilde{\ell}(\mu,\nu_i,X_{A,\Lambda}).
  \label{eq:ltildeBound}
\end{equation}
As $\ell$ is convex, $\tilde{\ell}$ is convex, as well \cite{B04}. If only witnesses are used, $\tilde{\ell}$ only has scalar variables, which results in a greatly enhanced numerical performance compared to the primal problem.

It is an interesting observation that $\tilde{\ell}$ is reminiscent of a grand canonical ensemble and
\be
\label{sepgibbs}
\tilde{\ell}=\ln Z'+\beta E=\ln\tr\exp(-\beta H')+\beta E\,,
\ee
where $H'=H-\frac{1}{\beta}\left(\sum_i \nu_i W_i + \sum_{A,\Lambda} \Lambda^*_A\otimes\1_{\overline{A}}[X_{A,\Lambda}]\right)$. Instead of particle numbers the constraints stem from the specific witnesses or maps used. We will sometimes refer to eq(\ref{sepgibbs}) as a \emph{witness canonical ensemble entropy} (WCEE).

In other words, when choosing $\mu=-\beta$, the operator $\sum_i \nu_i W_i + \sum_{A,\Lambda} \Lambda^*_A\otimes\1_{\overline{A}}[X_{A,\Lambda}]$ can be seen as some sort of grand canonical ensemble with additional "chemical" or rather \emph{witness potentials}. In particular the entropy gap can be lower bounded as follows:
\be\label{eq:EntropyGapBound}
\Delta S\ge\ln\frac{\tr\exp\left(-\beta H\right)}{\tr\exp\left(-\beta H+\sum_i \nu_i W_i + \sum_{A,\Lambda} \Lambda^*_A\otimes\1_{\overline{A}}[X_{A,\Lambda}]\right)}
\ee
for all $\nu_i\ge0$ and $X_{A,\Lambda}\ge0$.

It is also possible to compute $\ell$ for the linear entropy. After rewriting the Lagrangian as
\be
  L_\text{LIN} = 1 - \lambda - \mu E 
      + \tr\rho\left( \lambda\1 + \mu H + \nu_i W_i + X_0 +  \sum_{A,\Lambda} \Lambda^*_A\otimes\1_{\overline{A}}[X_{A,\Lambda}] \right) 
      - \tr\rho^2,
\ee
it is easy to show that the optimiser is given by
\be
  \rho^\text{crit} = \frac{1}{2}\left( \lambda\1 + \mu H + \nu_i W_i + X_0 +  \sum_{A,\Lambda} \Lambda^*_A\otimes\1_{\overline{A}}[X_{A,\Lambda}] \right).
\ee
Hence
\begin{equation}
  \ell_\text{LIN}(\lambda,\mu,\nu_i,X_0,X_{A,\Lambda})
         = \frac{1}{4}\tr\left( \lambda\1 + \mu H + \nu_i W_i + X_0 +  \sum_{A,\Lambda} \Lambda^*_A\otimes\1_{\overline{A}}[X_{A,\Lambda}]\right)^2
           + 1 - \lambda - \mu E.
  \label{eq:lin-ent-dual-fct}
\end{equation}
Minimisation of $\ell_\text{LIN}$ yields an upper bound on the maximal linear entropy achievable by separable states. Again, this minimisation is easier to deal with as in the von Neumann case, because it is only a quadratic function.

\subsection{On numerics}
Let us now briefly discuss how the primal and dual optimisation problems introduced in the preceding sections can be implemented numerically. As mentioned before, (\ref{eq:constrained-entropy}) and (\ref{eq:ltildeBound}) have concave and convex objectives, respectively, as well as linear and semidefinite constraints. Note that for our purposes it is sufficient to only compute the dual problems. It can however be instructive to also compute the primal problem in order to obtain the optimiser and check if strong duality holds. While interior-point methods can in principle solve such problems \cite{B04}, readily available solvers such as  Sedumi \cite{sedumi} or SDPT3 \cite{sdpt3} can only handle linear and quadratic objectives. This is sufficient to solve (\ref{eq:constrained-entropy}) and minimise (\ref{eq:lin-ent-dual-fct}) for the Linear entropy.

For the von Neumann entropy, however, there is a way to obtain an approximate solution: It is a well known fact that the von Neumann entropy of a state is equal to the Shannon entropy of its eigenvalues. As for $\tilde\ell$, note that for a hermitian $n\times n$ matrix $M$ with eigenvalues $\lambda_i$ it holds $\tr\exp(M)=\sum_i\exp(\lambda_i)$. Hence $\tilde\ell$ is a function of the eigenvalues of the exponent. Both $S$ and $\tilde\ell$ are invariant under permutation of the eigenvalues. This allows us to reformulate (\ref{eq:constrained-entropy}) as 
\begin{align}\label{eq:Eig_S}
S_{\text{max},\Lambda_A,W_i,E}=&\max_{\textbf{v},\rho} H(\textbf{v}),\\
& \text{s.t. }\textbf{v}=\text{eig}(\rho),\;\rho\text{ state},\;\tr\left(\rho H\right) = E,\nonumber\\
&\Lambda_A\otimes\1_{\overline{A}}[\rho]\geq 0,\tr(\rho W_i)\ge0\nonumber
\end{align}
where $\text{eig}(\rho)$ denotes the vector of eigenvalues of $\rho$. Note that this is not a semidefinite constraint. There is, however, a trick to include the eigenvalues into a semidefinite programme\cite{Loefberg-private}: While the eigenvalues are generally not SDP-representable, the sum of the $k$ largest eigenvalues of a matrix is \cite{alizadeh1995interior}. We can now replace the eigenvalue constraint in (\ref{eq:Eig_S}) by the constraint that $\textbf{v}$ has to majorise $\text{eig}(\rho)$, i.e.
\begin{align*}
&\tr(\rho)=\sum_iv_i\\
&s_1(\rho)\le v_1\\
&s_2(\rho)\le v_1+v_2\\
&...\\
&s_{n-1}(\rho)\le\sum_{i=1}^{n-1} v_i,
\end{align*}
where $s_k(M)$ denotes the sum of the k largest eigenvalues of a matrix $M$. Then the optimising $\textbf{v}$ will be equal to the eigenvalues of the optimising $\rho$. To see this, recall that if $\textbf{v}\succ\textbf{w}$, it holds $H(\textbf{v})\le H(\textbf{w})$ \cite{nielsen2001majorization}. Let us now assume that the (unique) optimising $\textbf{v}$ of the Shannon entropy majorises but is not equal to the eigenvalues of the optimal $\rho$. Then $S(\rho)$ would be greater or equal to $H(\textbf{v})$, which is a contradiction.
Since $\textbf{v}\succ\textbf{w}$ implies $f(\textbf{v})\ge f(\textbf{w})$ for any convex function\cite{nielsen2001majorization}, the same argument can be applied to (\ref{eq:ltildeBound}).

So far, we have replaced the von Neumann by the Shannon entropy and, in the dual problem, the matrix exponential by scalar exponential functions and transformed the arising eigenvalue constraints into semidefinite constraints. In order to apply Sedumi or SDPT3, all that is left to do is to approximate the objectives by piecewise linear functions. While this will only give us approximate solutions, let us note that the optimisers found in this way, can be easily checked to be feasible and inserted in the original objective.

When using PPT or other positive maps the dimension of the matrix variables increases exponentially with the number of qubits. As we will show in the next section, this makes it difficult to go beyond five qubits. The computation of (\ref{ell}) greatly simplifies when only witnesses are used. In this case there will only be scalar variables $\m$ and $\n_i$ and no semidefinite constraints. This allows for application of a standard non-linear solver, such as FMINCON \cite{mathworksfmincon}. Using FMINCON we were able to obtain results for up to 13 qubits, as we will present in the next section.

If we want to show if the Gibbs state is entangled, we will also need to compute the Gibbs state entropy $S(\rho_{Gibbs})=\beta E+\ln Z$. To do so numerically it is sufficient to compute the eigenvalues of the Hamiltonian, which allows for computation of the partition function $Z=\tr e^{-\beta H}$ and the mean energy $E=\tr\rho_{Gibbs} H=\frac{1}{Z}\sum_i e^{-\beta E_i}E_i$.

\subsection{Examples}
In order to test our method, we have implemented it for \textit{Heisenberg model}, which was introduced in order to simplify the analysis of systems of spins, such as ferro- or antiferromagnets. It only takes into account the nearest neighbour exchange interaction between the spins as well an external magnetic field. In the one dimensional case, i.e. a chain of spins, the Heisenberg model is described by the following Hamiltonian
\begin{equation}
H=-\sum_{i=1}^N\left(J_x\s^x_i\s^x_{i+1}+J_y\s^y_i\s^y_{i+1}+J_z\s^z_i\s^z_{i+1}\right)+B\sum_{i=1}^N\s^z_i
\end{equation} 
where $J_x,J_y,J_z$ are the coupling constants for the $x,y,z$-components of the spins, $N$ the number of spins and $B$ the external magnetic field in $z$direction. $\s^x_i,\s^y_i,\s^z_i$ denote the Pauli operators for the $i$-th spin. Let us assume periodic boundary conditions, i.e. a ring of spins. If $J_x=J_y=J_z=:J$, we are talking about an \textit{isotropic XXX Heisenberg model}. $J>0$ and $J<0$ correspond to ferromagnetic and antiferromagnetic systems, respectively.  If $J_x=J_y\ge J_z$, the system is called an \textit{XXZ} system and so on. If only one component of the spin is considered, i.e. only one $J_i\neq 0$, the Heisenberg model reduces to the \textit{Ising model}.

The numerical results presented below have been obtained using either Sedumi or SDPT3 as well as Yalmip \cite{YALMIP}. Using PPT constraints we have applied the method for up to five qubits, for witnesses for up to 13 qubits.

\subsubsection{Antiferromagnetic Heisenberg model}
For the one dimensional antiferromagnetic Heisenberg model (XXX with $J=-1$), it is possible to detect entanglement at higher energies than \cite{dowling2004energy,BV04,toth2005entanglement}, both using the von Neumann and the linear entropy as proxies. The computations have been performed using both the partial transpose with respect to the partition $A_{i\text{ even}}|A_{j\text{ odd}}$ and all possible partitions.

Let us start with the results for the even versus odd partition. In tables \ref{table1} the energy ranges where entanglement is detected are shown for the von Neumann entropy and for the linear entropy, respectively. Here $N$ denotes the number of qubits, $E_0$ the ground state energy and $E_\text{min,PPT-even-odd}$ denotes the smallest mean energy allowing for PPT w.r.t the even versus odd partition. Below that mean energy any state is guaranteed to be entangled. Note that $E_\text{min,sep}\ge E_\text{min,PPT-all}\ge E_\text{min,PPT-even-odd}$. Between $E_\text{min,PPT-even-odd}$ and $E_\text{max,gap}$ all states falling into the entropy gap are entangled. I.e. $E_\text{max,gap}$ is the largest energy for which the proxy method can work. This includes the Gibbs state in the von Neumann case. Note that the Gibbs state is not necessarily the maximiser for the linear entropy. Hence, instead of the Gibbs state we have computed the state with maximal linear entropy 
\be
S_{\text{L,max}}=\max_{\rho\text{ state},\;\tr H\rho= E}S_L(\rho)
\ee
and compared it to the constraint linear entropy. $\cal E$ is defined as the fraction of the energy range where entanglement is detected
\be
\mathcal{E}=\frac{E_\text{max,gap}-E_0}{E_\text{max}-E_0},
\ee
where $E_\text{max}$ denotes the largest energy eigenvalue.

		\begin{table}
		\begin{center}
	
	\begin{tabular}{|l|l|l|l|l|}\hline
$N$	&$E_0/N$ & $E_\text{min,PPT-even-odd}/N$&$E_\text{max,gap}/N$&$\cal E$\\\hline
3&-1.000&-1.000&-0.610&0.195\\\hline
4&-2.000&-1.000&-0.660&0.447\\\hline
5&-1.494&-1.008&-0.695&0.320\\\hline
\end{tabular}
\caption{\label{table1}Energy ranges in which the von Neumann entropy can be successfully used a proxy for detecting entanglement in the antiferromagnetic Heisenberg model ($J=-1$).}
	\begin{tabular}{|l|l|l|l|l|}\hline
$N$	&$E_0/N$ & $E_\text{min,PPT-even-odd}/N$&$E_\text{max,gap}/N$&$\cal E$\\\hline
3&-1.000&-1.000&-0.600&0.200\\\hline
4&-2.000&-1.000&-0.370&0.543\\\hline
5&-1.494&-1.008&-0.302&0.478\\\hline
\end{tabular}
\caption{Energy ranges in which the linear entropy can be successfully used a proxy for detecting entanglement in the antiferromagnetic Heisenberg model ($J=-1$).}
\end{center}
\end{table}

As can be seen it tables \ref{table1}, entanglement can be detected in a large area of the energy spectrum. The maximum mean energy where entanglement can be detected is also substantially higher than $-N$, which is the maximum mean energy where the methods of \cite{dowling2004energy,BV04,toth2005entanglement} work. Let us also note that the ground state energies given in \cite{dowling2004energy,BV04,toth2005entanglement} are only correct in the limit of large $N$, as has been noted in \cite{toth2005entanglement}. The fact that the linear entropy detects more than the von Neumann entropy could be a result of the piecewise linear approximation of the von Neumann entropy in the maximisation, which results in weaker bounds on $S_\text{max,sep}$.

The entropy gap can be seen in figure \ref{plot:PPT5}a and \ref{plot:PPT5}b for the von Neumann entropy and for the linear entropy, respectively. Here the Gibbs state entropy or $S_{\text{L,max}}$, as well as $S_\text{max,PPT-even-dd}$ are plotted versus the mean energy. 
\begin{figure}
	\centering
	(a)\includegraphics[width=0.4\textwidth]{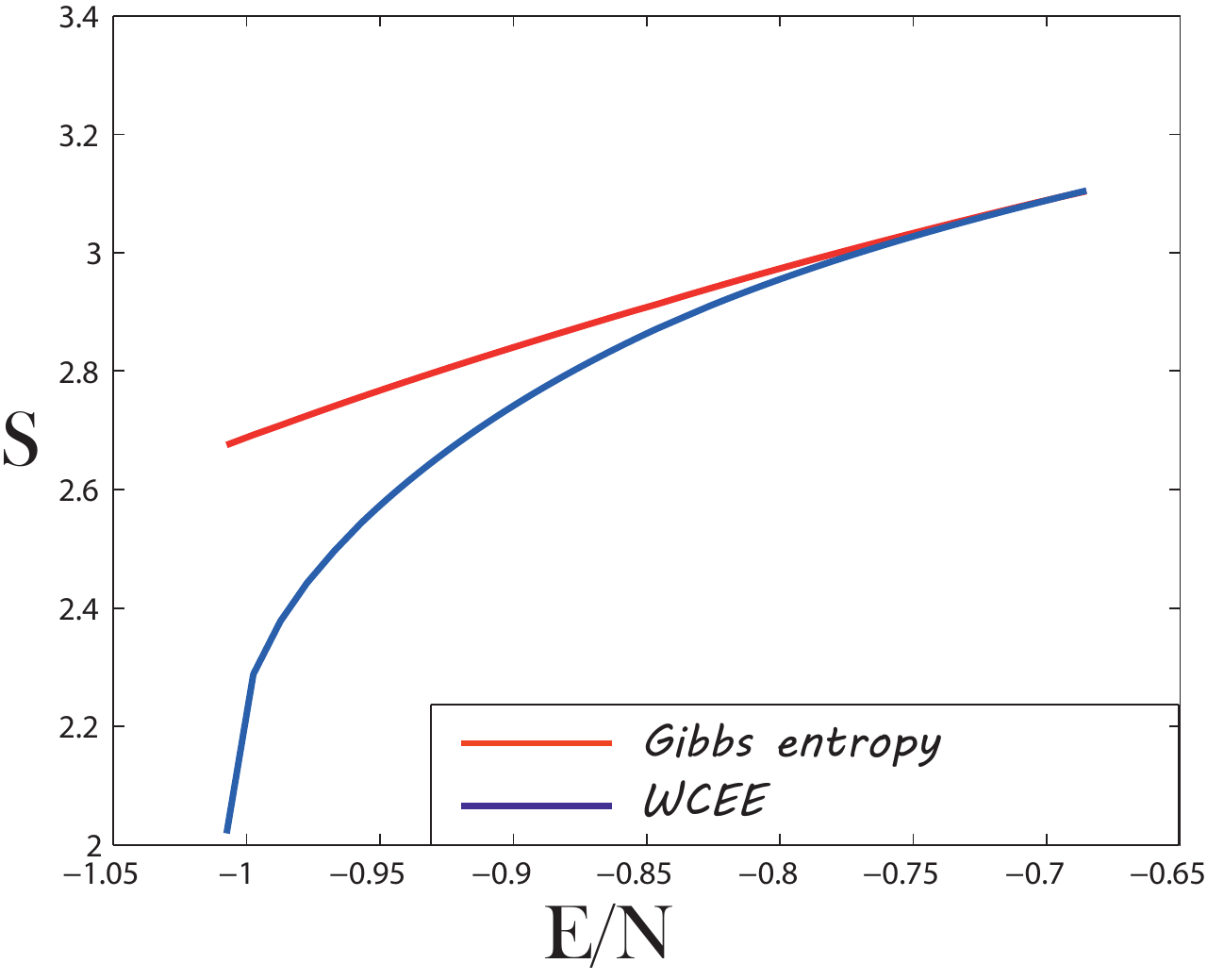}
	(b)\includegraphics[width=0.4\textwidth]{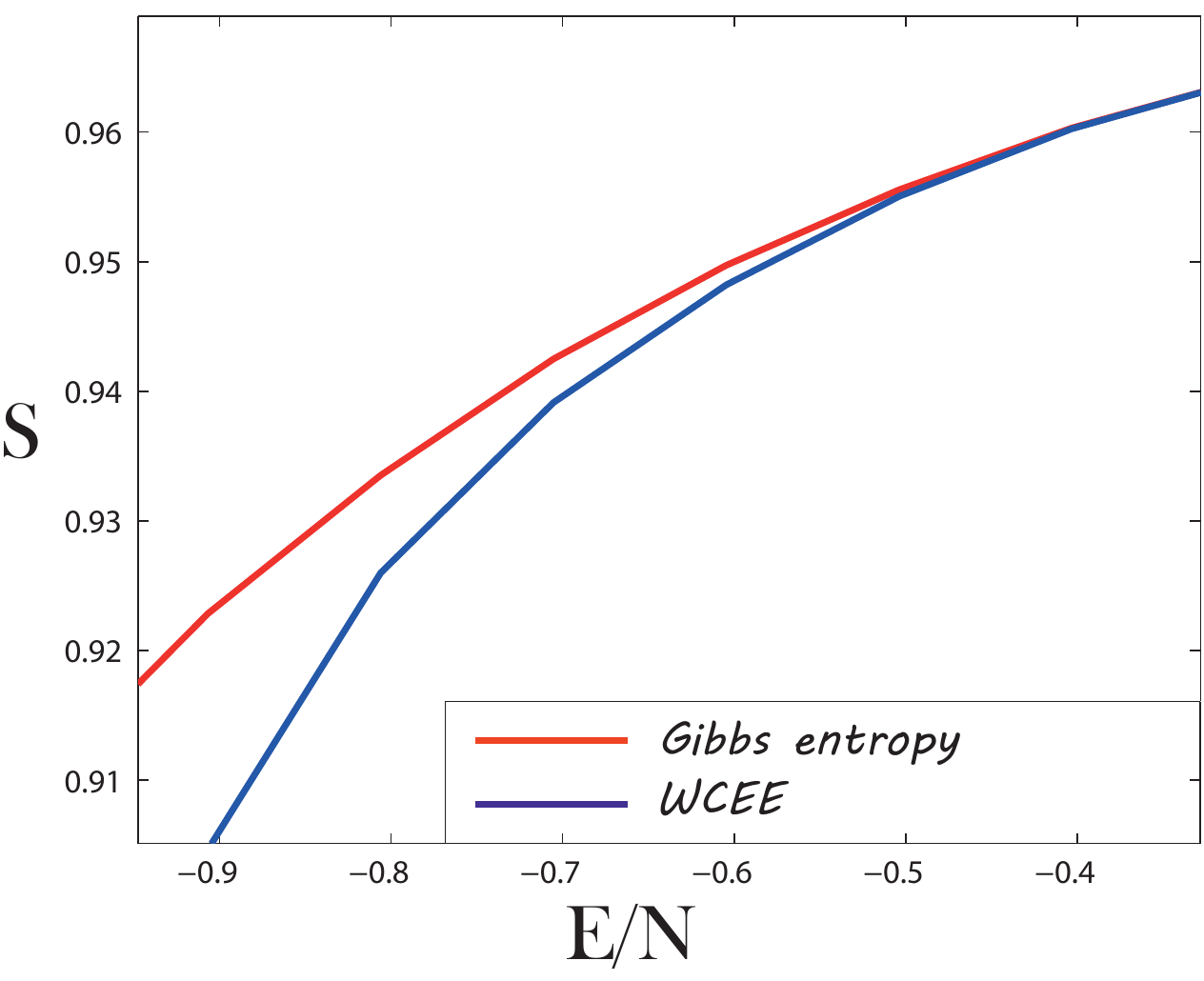}
	\caption{\label{plot:PPT5}The entropy gap for the antiferromagnetic Heisenberg model ($J=-1$) with $5$ qubits and the PPT condition w.r.t. subsystems $1,3,5$. (a) shows the von Neumann entropy, (b) the linear entropy.
	The Gibbs state entropy (red line) and the dual of the constraint, i.e. the \emph{witness canonical ensemble entropy} (WCEE), (blue line) plotted versus the mean energy. The difference of the two lines constitutes the \emph{entropy gap}.}
\end{figure}
Entanglement can also be detected in the presence of a magnetic field. Tables \ref{table3} and \ref{table4} show the detected energy ranges in a three qubit system for the von Neumann and linear entropies, respectively. The entropies are plotted in figure \ref{plot:B} for $B=3$.

	\begin{table}
		\begin{center}
	\begin{tabular}{|l|l|l|l|l|}\hline
$B$	&$E_0/N$ & $E_\text{min,PPT-even-odd}/N$&$E_\text{max,gap}/N$&$\cal E$\\\hline
0&-1.000&-1.000&-0.610&0.195\\\hline
1&-1.333&-1.333&-0.720&0.184\\\hline
2&-1.667&-1.667&-1.033&0.136\\\hline
3&-2.000&-2.000&-1.540&0.077\\\hline
\end{tabular}
\caption{\label{table3}Energy ranges for three qubits in the antiferromagnetic Heisenberg model ($J=-1$) in a magnetic field $B=3$ using von Neumann entropy}
\end{center}
\end{table}

		\begin{table}
		\begin{center}
	\begin{tabular}{|l|l|l|l|l|}\hline
$B$	&$E_0/N$ & $E_\text{min,PPT-even-odd}/N$&$E_\text{max,gap}/N$&$\cal E$\\\hline
0&-1.000&-1.000&-0.600&0.200\\\hline
1&-1.333&-1.333&-0.560&0.232\\\hline
2&-1.667&-1.667&-0.717&0.204\\\hline
3&-2.000&-2.000&-0.940&0.177\\\hline
\end{tabular}
\caption{\label{table4}Energy ranges for three qubits in the antiferromagnetic Heisenberg model ($J=-1$) with a magnetic field $B=3$ using linear entropy}
\end{center}
\end{table}

\begin{figure}
	\centering
	(a)\includegraphics[width=0.4\textwidth]{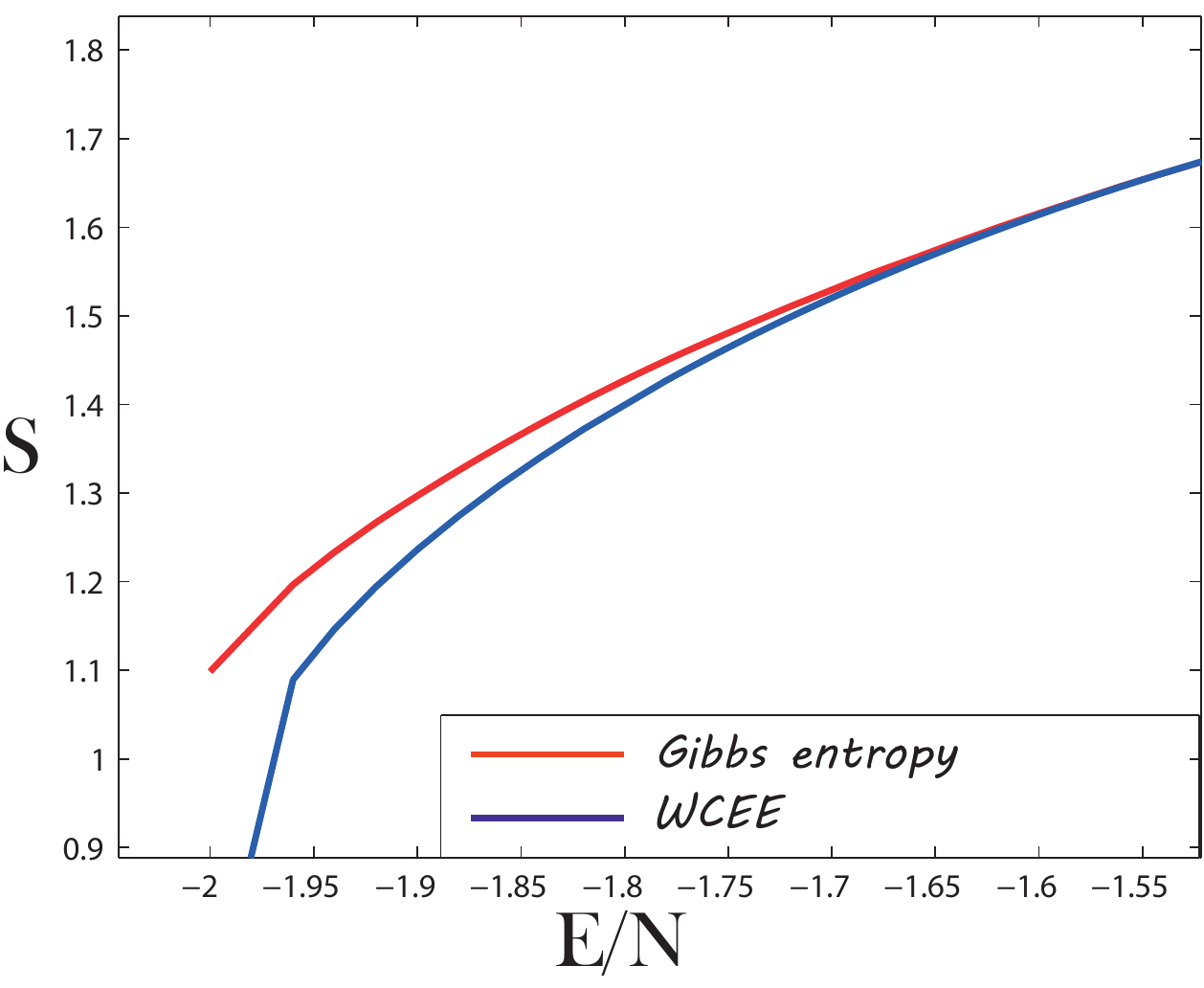}
	(b)\includegraphics[width=0.4\textwidth]{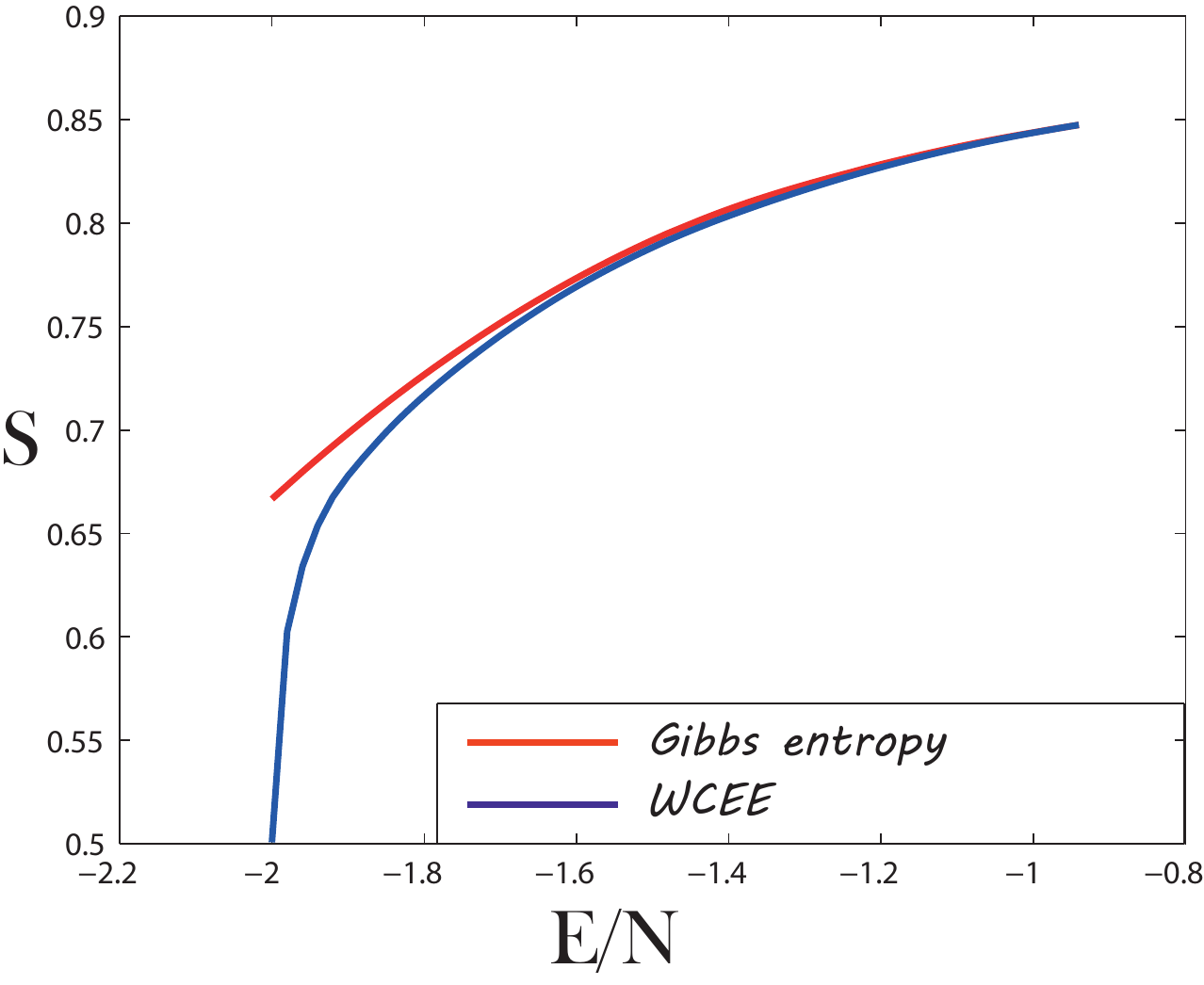}
	\caption{\label{plot:B}The entropy gap for 3 qubits in the antiferromagnetic Heisenberg model ($J=-1$) with a magnetic field $B=3$, again for PPT w.r.t. the even versus uneven partition. (a) shows the von Neumann entropy, (b) the linear entropy.}
\end{figure}
Going from PPT w.r.t. the even-uneven partition to all possible partitions greatly increases the computation time as the number of partitions grows exponentially with the number of qubit. Since the constrained entropy increases monotonically in $E$, it holds $E_\text{max,gap,PPT-all}\ge E_\text{max,gap,PPT-even-odd}$, possibly increasing the energy range, where entanglement can be detected. In the example considered here, however, only a small increase can be obtained. See tables \ref{table5} and \ref{table6}, as well as figure \ref{plot:all}. For three qubits it holds $E_\text{min,PPT-all}/N=-0.6>-1$, implying that $E_\text{min,sep}/N\ge-0.6>-1$. This shows that the result of \cite{dowling2004energy,toth2005entanglement} is suboptimal for an odd number of qubits, as mentioned in \cite{toth2005entanglement}. The reason is that they use a partition into two sub-lattices such that every neighbouring sites belong to different sub-lattices, which is not possible for an odd number of qubits. 
		\begin{table}
		\begin{center}
	\begin{tabular}{|l|l|l|l|l|}\hline
$N$	&$E_0/N$ & $E_\text{min,PPT-all}/N$&$E_\text{max,gap}/N$&$\cal E$\\\hline
3&-1.000& -0.600	&-0.600&0.200\\\hline
4&-2.000&-1.000	&-0.660  &0.447 \\\hline
5&-1.494&-0.809&-0.809 &0.275 \\\hline
\end{tabular}
\caption{\label{table5}Results for the von Neumann entropy in the antiferromagnetic Heisenberg model ($J=-1$) with a magnetic field $B=3$ using all possible partitions}
\end{center}
\end{table}

		\begin{table}
		\begin{center}
	\begin{tabular}{|l|l|l|l|l|}\hline
$N$	&$E_0/N$ & $E_\text{min,PPT-all}/N$&$E_\text{max,gap}/N$&$\cal E$\\\hline
3&-1.000&-0.600&	-0.594&0.203\\\hline
4&-2.000&-1.000&	-0.360 &0.547 \\\hline
\end{tabular}
\caption{\label{table6}Results for the linear entropy in the antiferromagnetic Heisenberg model ($J=-1$) with a magnetic field $B=3$ using all possible partitions}
\end{center}
\end{table}

\begin{figure}
	\centering
	(a)\includegraphics[width=0.4\textwidth]{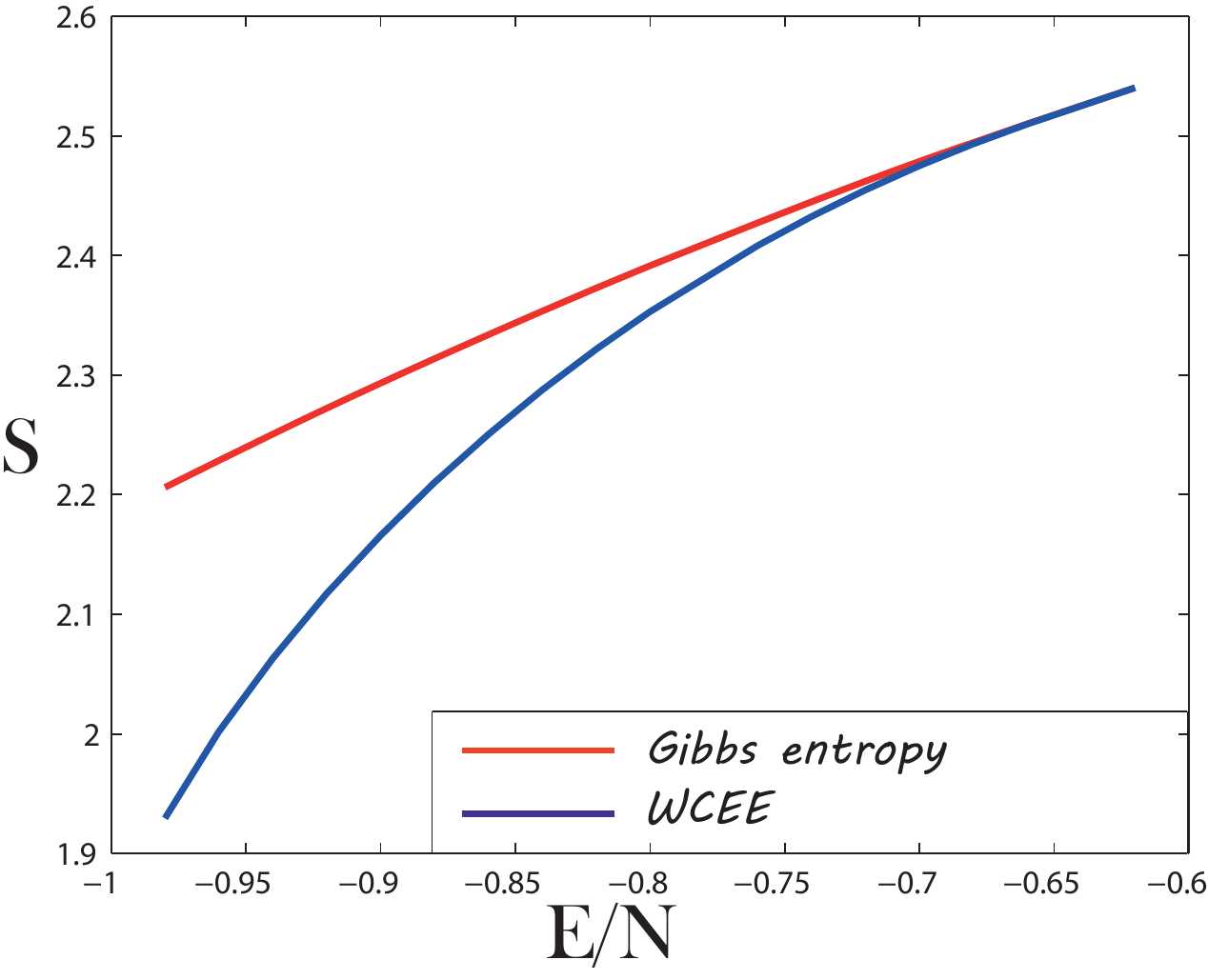}
	(b)\includegraphics[width=0.4\textwidth]{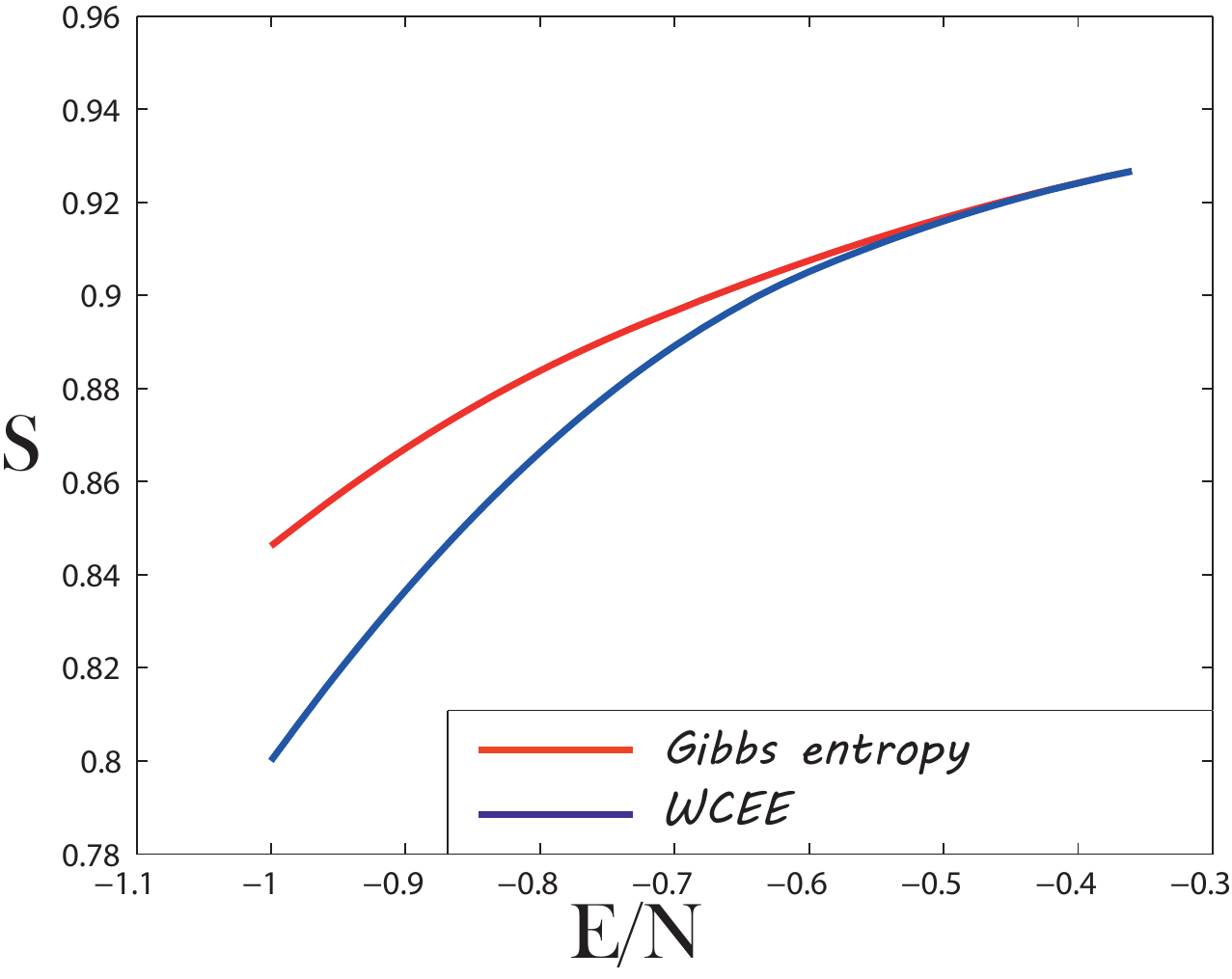}
	\caption{\label{plot:all}The entanglement gap for all possible partitions for the von Neumann (a) and linear (b) entropies in the antiferromagnetic Heisenberg model ($J=-1$) with a magnetic field $B=3$.}
\end{figure}

\subsubsection{Dicke states}
Dicke states were first considered in the theory of coherent spontaneous light emission \cite{PhysRev.93.99}. They are defined by
\be
\ket{D^n_m}=\binomial{n}{m}^{-\frac{1}{2}}\sum_\a\ket{d_\a},
\ee
where 
\be
\ket{d_\a}=\bigotimes_{i\notin\a}\ket{0}_i\bigotimes_{i\in\a}\ket{1}_i,
\ee
m is the number of excitations, $\a$ denote sets of indices of excited subsystems and the sum is taken over all inequivalent sets of $m$ indices.

More recently Dicke states turned out to be a useful resource for quantum information processing task as they are LOCC transferable to GHZ or W states \cite{kiesel2007experimental}. Several experiments have successfully created Dicke states, e.g. \cite{kiesel2007experimental,wieczorek2009experimental,prevedel2009experimental}.

In \cite{zhou2011ground}, it has been shown using perturbation theory that the approximate ground states of anisotropic ferromagnetic XXZ Heisenberg Hamiltonians (i.e. $J_x=J_y\ge J_z>0$) are Dicke states. For $B=0$ and $J_x=J_z>0$ the ground states are given by $\{\ket{D^n_m}\}_{m=0}^n$ with $n+1$-fold degeneracy, which makes entanglement at low energies unlikely. This is because the system will soon go into a mixture of the Dicke states, which is separable. For $J_x>J_z$ and $B\neq0$, however, the degeneracy vanishes. In first order perturbation theory $\ket{D^n_m}$ is the ground state for 
\begin{equation}\label{Bbounds}
-\frac{n-2m+1}{n-1}\Delta J<B<-\frac{n-2m-1}{n-1}\Delta J
\end{equation}
where $\Delta J=J_x-J_z$ is the anisotropy parameter.

A number of methods have been developed to detect genuinely multipartite entanglement in Dicke states \cite{toth2007detection,PhysRevLett.99.193602,campbell2009characterizing,PhysRevLett.103.100502,huber2011experimentally,Dickeref1,Dickeref2}, of which \cite{huber2011experimentally} is the most generally applicable one. It consists of a non-linear witness, i.e. an non-linear inequality that has to hold for every biseprable state. In order apply the witness only a polynomial (in the number of qubits) number of local measurements is necessary.

Combining \cite{huber2011experimentally} with the use of proxy witnesses, entanglement can be detected even more easily experimentally. To do so, let us introduce a (weaker) linear version of the non-linear witness given in \cite{huber2011experimentally}. It is given by
\begin{equation}\label{DickWit}
W^n_m=\frac{1}{2}\sum_{(\a,\beta)\in\g}\left(-\ket{d_\beta}\bra{d_\a}-\ket{d_\a}\bra{d_\beta}+\pro{d_{\a\cap\beta}}+\pro{d_{\a\cup\beta}}\right),
\end{equation}
where $\g=\left\{(\a,\beta):\left|\a\cap\beta\right|=m-1\right\}$. $W^n_m$ detects $\ket{D^n_m}$ in a maximal way. Using these witnesses as a constraint in (\ref{eq:constrained-entropy}), it is possible to detect entanglement for different choices of $J_x$ and $J_z$ and $B$ satisfying (\ref{Bbounds}). Results for 11 and 13 qubits are presented in table \ref{table7} and figures \ref{plot:DickeE} and \ref{plot:DickeT}.

		\begin{table}
		\begin{center}
	\begin{tabular}{|l|l|l|l|l|l|l|l|l|l|}\hline
$N$&	$m$&	$B$	&	$\Delta J$&	$E_0/N$	&$E_\text{min,gap}/N$	&$E_\text{max,gap}/N$&	$T_\text{min,gap}$	&$T_\text{max,gap}$&	$\cal E$\\\hline
11&	5	&-1	&10	&-13.753	&-13.599&	-12.890	&3.970&	7.930&	0.031\\\hline
11&	5	&-3		&20	&-26.707	&-26.299&	-24.698	&8.460&	16.420&	0.038\\\hline
11&	5	&-1		&20	&-26.525&	-26.344	&-24.802	&6.940&	15.850&	0.033\\\hline
13&	6	&-1		&12	&-16.276	&-16.149&	-15.577&	3.970&	7.930	&0.021\\\hline
\end{tabular}
\caption{\label{table7}Detection energies and entropies for the XXZ Hamiltonian with $J_z=1$ using (\ref{DickWit}) as witness.}
\end{center}
\end{table}

\begin{figure}

	\centering
	(a)\includegraphics[width=0.4\textwidth]{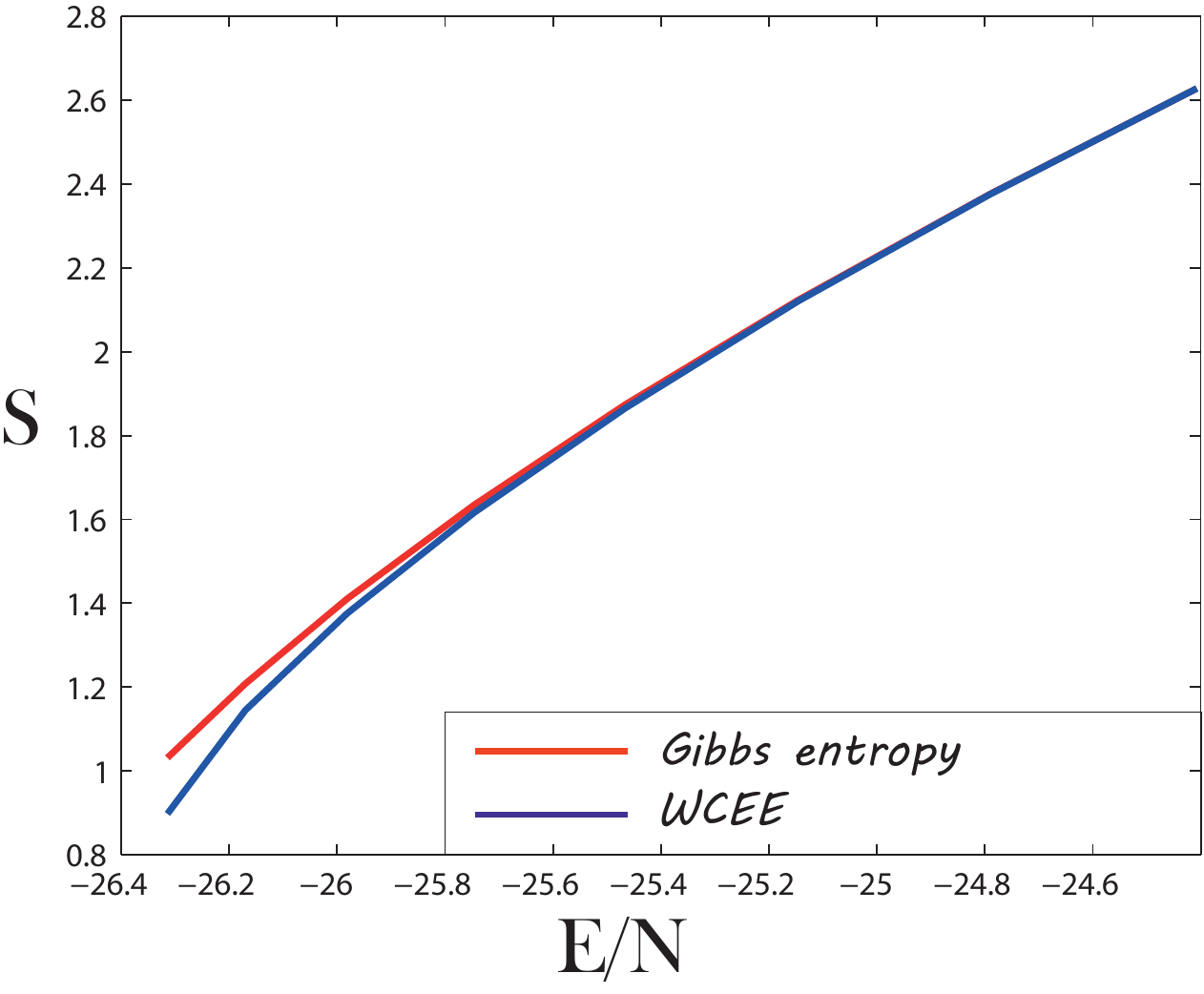}
	(b)\includegraphics[width=0.4\textwidth]{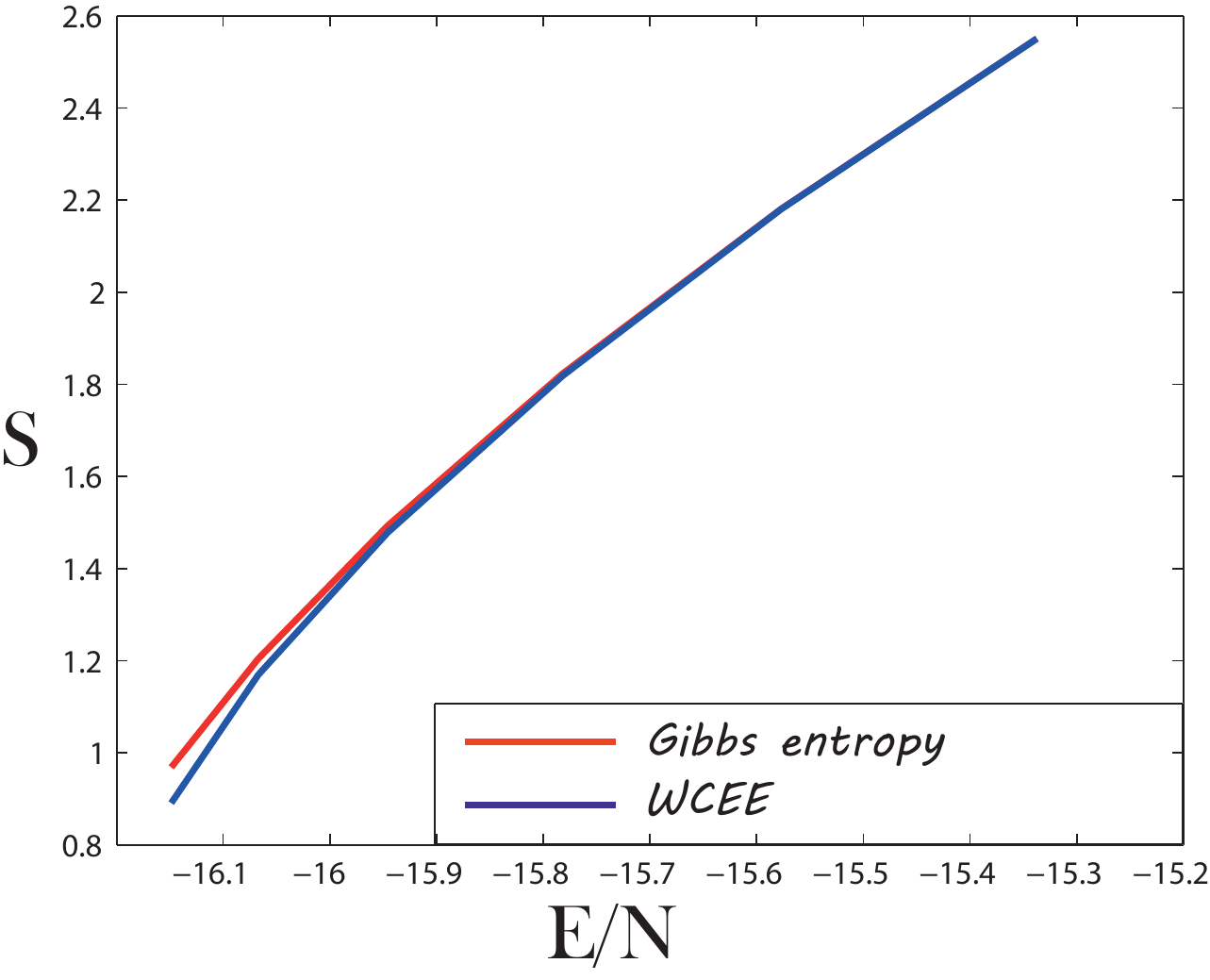}
	\caption{\label{plot:DickeE}Entropy gap plotted versus mean energy for (a) 11 qubits, $m=5$, $B=-3$ and $\Delta J=20$ and (b)  13 qubits, $m=6$, $B=-1$ and $\Delta J=12$.}
\end{figure}

\begin{figure}

	\centering
	(a)\includegraphics[width=0.4\textwidth]{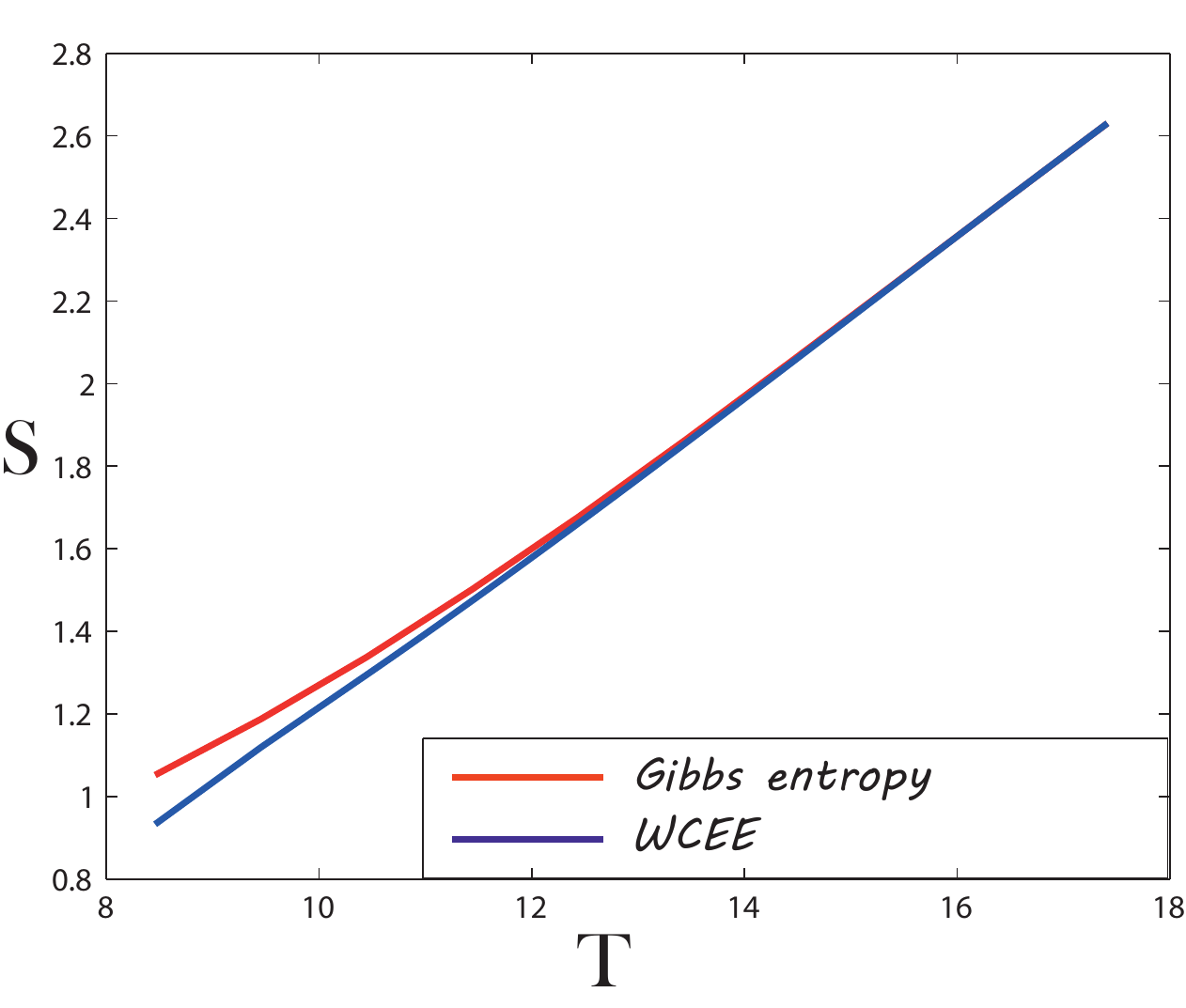}
	(b)\includegraphics[width=0.4\textwidth]{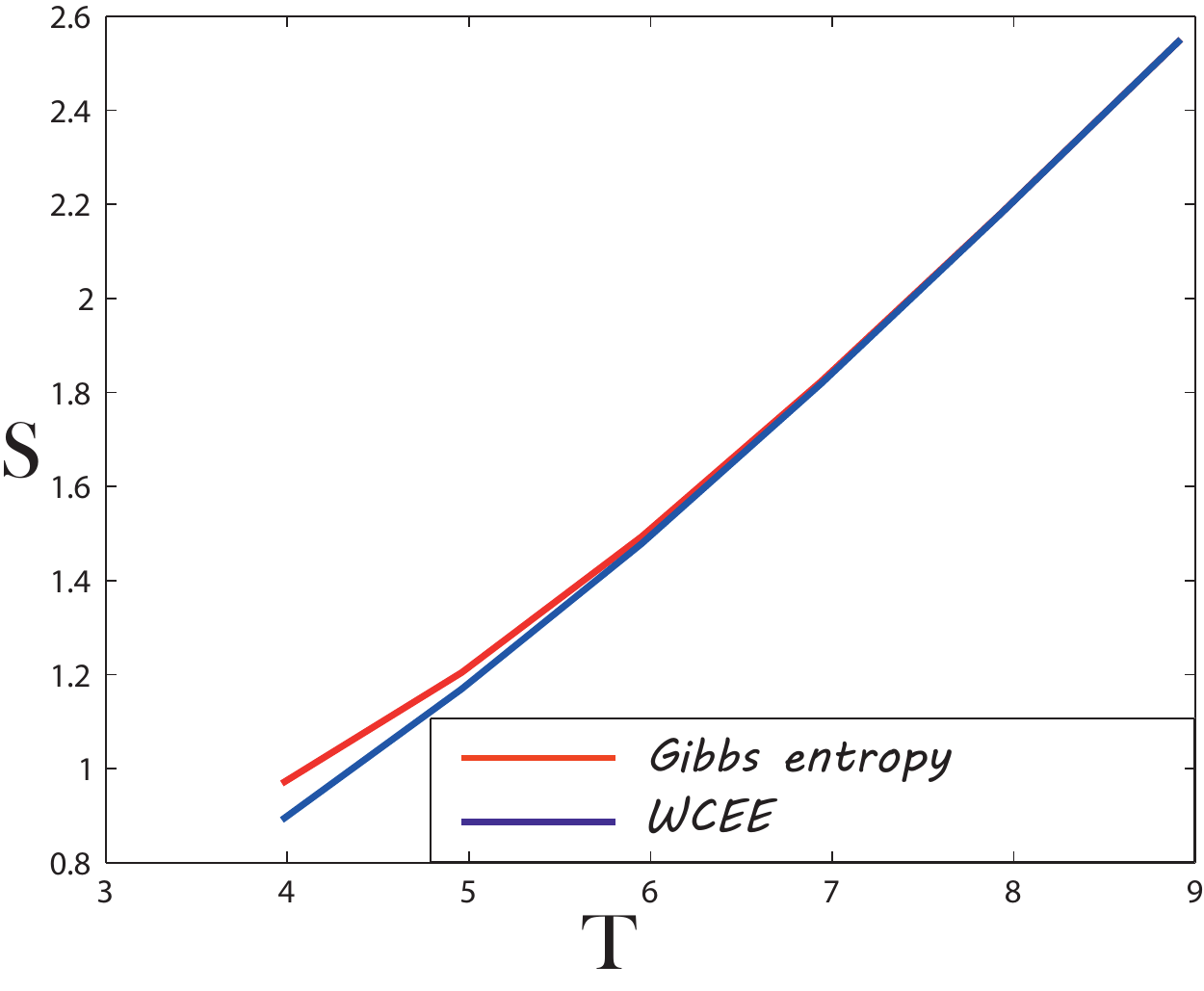}
	\caption{\label{plot:DickeT}Entropy gap plotted versus temperature for (a) 11 qubits, $m=5$, $B=-3$ and $\Delta J=20$ and (b)  13 qubits, $m=6$, $B=-1$ and $\Delta J=12$.}
\end{figure}

\section{To the thermodynamic limit}
\subsection{The ground state witness}
In order to extend the above results to the thermodynamic limit, let us now focus on witnesses, which are diagonal in the energy eigenbasis. A general example of such a witness is given as $W=\a\1-\pro{E_0}$, where $\ket{E_0}$ is the ground state of $H$ and $\a=\max_{\ket{\phi}, \ket{\varphi}}\left|\bra{E_0}\ket{\phi}\ket{\varphi}\right|^2$ \cite{guhne2009entanglement}. Instead of $\pro{E_0}$, one could also use any projector onto one or several energy eigenstates. The only obvious requirement being that the states are entangled themselves, otherwise the resulting operators would not be witnesses. We can then either maximise over a fixed bipartition, or over all possible ones. In the former case it is possible to detect bipartite entanglement w.r.t. the partition chosen, in the latter case we can detect genuinely multipartite entanglement. Inserting the witness into (\ref{ell}), we see that the exponent becomes diagonal.
\begin{align}
\tilde{\ell}(\mu,\nu)&=\ln\tr\exp\left(\sum_i\left(\m E_i+\n(\a-\delta_{i0})\right)\pro{E_i}\right)-\m E\\
&=\ln\sum_{i}\exp\left(\m E_i+\n(\a-\delta_{i0})\right)-\m E,
\end{align}
where $E_i$ and $\ket{E_i}$ are the energy eigenvalues and eigenvectors, respectively and $\delta_{i0}$ is the Kronecker delta. Setting the gradient equal to zero, we obtain
\begin{align*}
\frac{\partial}{\partial\m}\tilde{\ell}(\mu,\nu)=0\Rightarrow &(i)\;\sum_i\exp\left(\m E_i+\n(\a-\delta_{i0})\right)(E_i-E)=0\\
\frac{\partial}{\partial\n}\tilde{\ell}(\mu,\nu)=0\Rightarrow &(ii)\;\sum_i\exp\left(\m E_i+\n(\a-\delta_{i0})\right)(\a-\delta_{i0})=0,
\end{align*}
which, by convexity of $\tilde{\ell}$, are sufficient conditions for a minimum. Solving those transcendent equations w.r.t. the Lagrange variables is only possible numerically for small $N$. Still it is possible to obtain a result for the thermodynamic limit, as we will now demonstrate.

Let us first consider the case where $\a=1$, i.e. $H$ has a separable ground state. Since the first term in the sum of (ii) vanishes, (ii) cannot be fulfilled for any $\n\ge0$. Since $\tilde{\ell}$ increases with $\n$, the minimum is attained at $\n=0$. At $\n=0$, however, (\ref{ell}) is an upper bound on the unconstrained problem, namely
\[
\min_\m\ln\sum_i\exp(\m E_i)-\m E\ge S(\rho_{Gibbs})
\]
Since $S(\rho_{Gibbs})=\ln Z+\beta E$, the minimum is attained at $\m=-\beta$. In this case no entanglement is detected.

Let us now move on to the case where $\a<\frac{e^{-\beta E_0}}{Z}$. Assuming that the minimum is attained at $\n=0$, i.e. that (ii) is fulfilled at $\n\le0$, (ii) tells us that 
\[
\a=\frac{\exp(-\beta E_0+\n(\a-1))}{\sum_i\exp(-\beta E_i+\n(\a-\delta_{i0}))}=\frac{e^{-\beta E_0}}{\sum_i e^{-\beta E_i} e^{\n(1-\delta_{i0})}}=\frac{e^{-\beta E_0}}{e^{-\beta E_0}+\sum_{i>0}e^{-\beta E_i}e^\n}\ge\frac{e^{-\beta E_0}}{Z},
\]
with equality for $\n=0$, which is a contradiction. Hence the minimum is attained at $\n>0$. This implies that $\min_{\n\ge0,\m}\tilde{\ell}<\min_\mu\tilde{\ell}(0,\m)=S(\rho_{Gibbs})$, hence $S_{\text{max},\Lambda_A,W_i,E}<S(\rho_{Gibbs})$. 
\begin{theorem}\label{alphaResult}
Let $\SEP$ be the set of separable pure states w.r.t. some partitions and $\a=\max_{\ket{\phi}\in\SEP}\left|\bra{E_0}\ket{\phi}\right|^2$. Then, if $\a<\frac{e^{-\beta E_0}}{Z}$, the Gibbs state as well as any state with entropy $S_{\text{max},\Lambda_A,W_i,E}<S<S(\rho_{Gibbs})$ will be inseparable w.r.t. those partitions.
\end{theorem}
Note that while for $\a<\frac{e^{-\beta E_0}}{Z}$ entanglement of the Gibbs state can also be detected by applying the witness directly, the same is not true for states with entropy $S_{\text{max},\Lambda_A,W_i,E}<S<S(\rho_{Gibbs})$. Since such states always exist, this shows that our method generically works. Theorem \ref{alphaResult} can be applied in the thermodynamic limit. Since $E_0$ generally scales linear in $N$ while $Z$ does so exponentially, Theorem \ref{alphaResult} can yield non-trivial results where $\a=\mathcal{O}(e^{-N})$. 
\subsection{Example}
We consider a special case of the Heisenberg Hamiltonian, given by
\[
H=-\sum_{i=1}^N\left(\frac{1+r}{2}\s^x_i\s^x_{i+1}+\frac{1-r}{2}\s^y_i\s^y_{i+1}+h\s^z_i\right).
\]
This is known as the \textit{XY} model in a transverse magnetic field. If $r=1$, we talk about the \textit{Ising} model, if $r=0$, the \textit{XX} model. The XY model undergoes a phase transition at $h=1$ \cite{osborne2002entanglement}. The ground state energy as well as the partition function have been computed in the thermodynamic limit by \cite{katsura1962statistical}. According to a conjecture numerically tested in \cite{wei2005global}, $\a$ w.r.t. full separability is given by
\[
\lim_{N\to\infty}\frac{\ln\a}{N}=2\max_{\xi}\int_0^{\frac{1}{2}}d\m \ln \left|\cos\theta\cos^2\frac{\xi}{2}+\sin\theta\sin^2\frac{\xi}{2}\cot\pi\m\right|,
\]
where $\tan 2\theta=\frac{r\sin 2\pi\m}{h-\cos{2\pi\m}}$ and $-\frac{\pi}{2}\le\theta\le\frac{\pi}{2}$. Using this, we can numerically show entanglement for a wide range of $T$ and $h$. See Figure \ref{XY_fig}. In particular we can show entanglement in vicinity of the phase transition, in accordance with \cite{osborne2002entanglement}. 

\begin{figure}

	\centering
	(a)\includegraphics[width=0.6\textwidth]{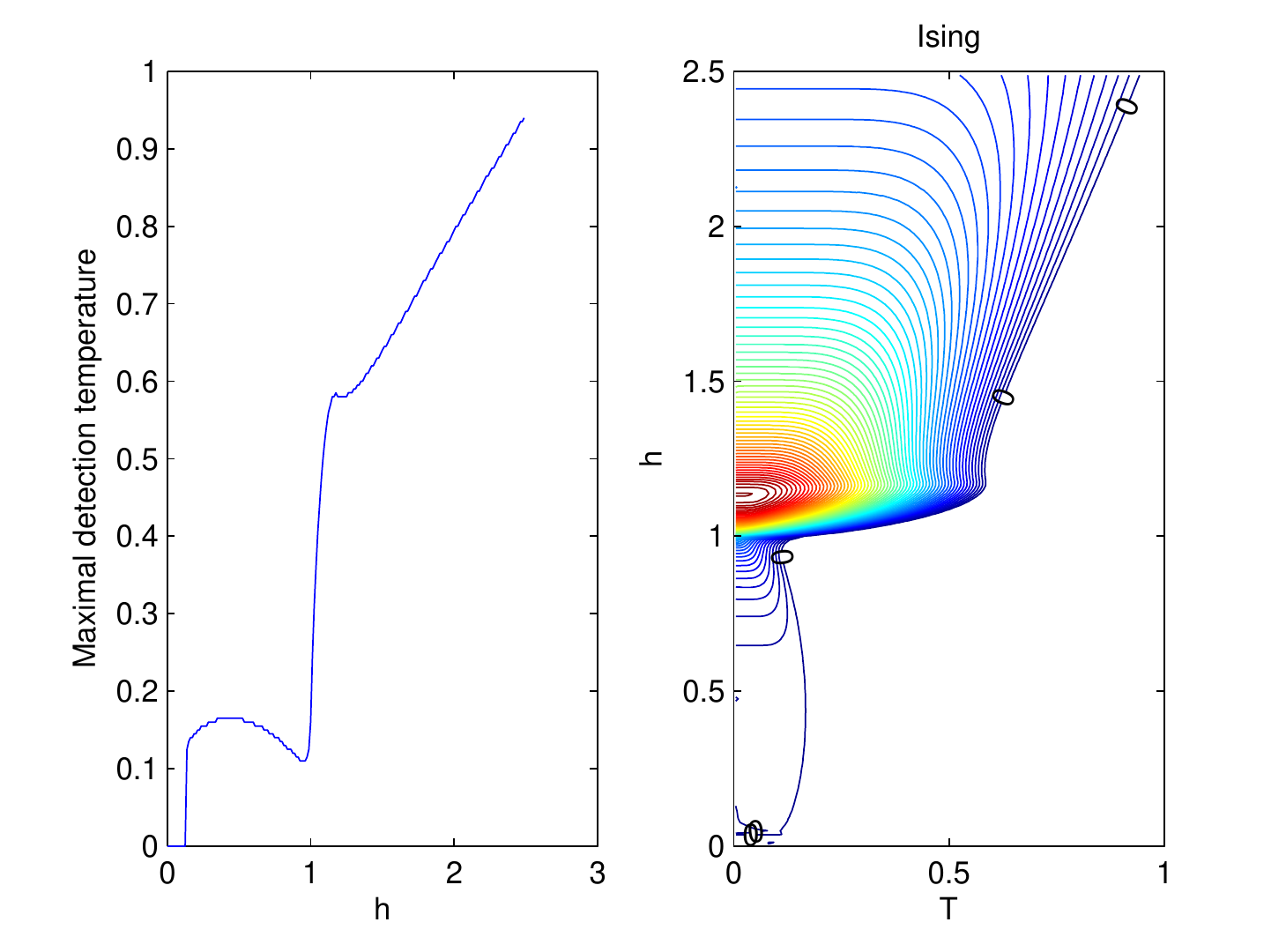}
	(b)\includegraphics[width=0.6\textwidth]{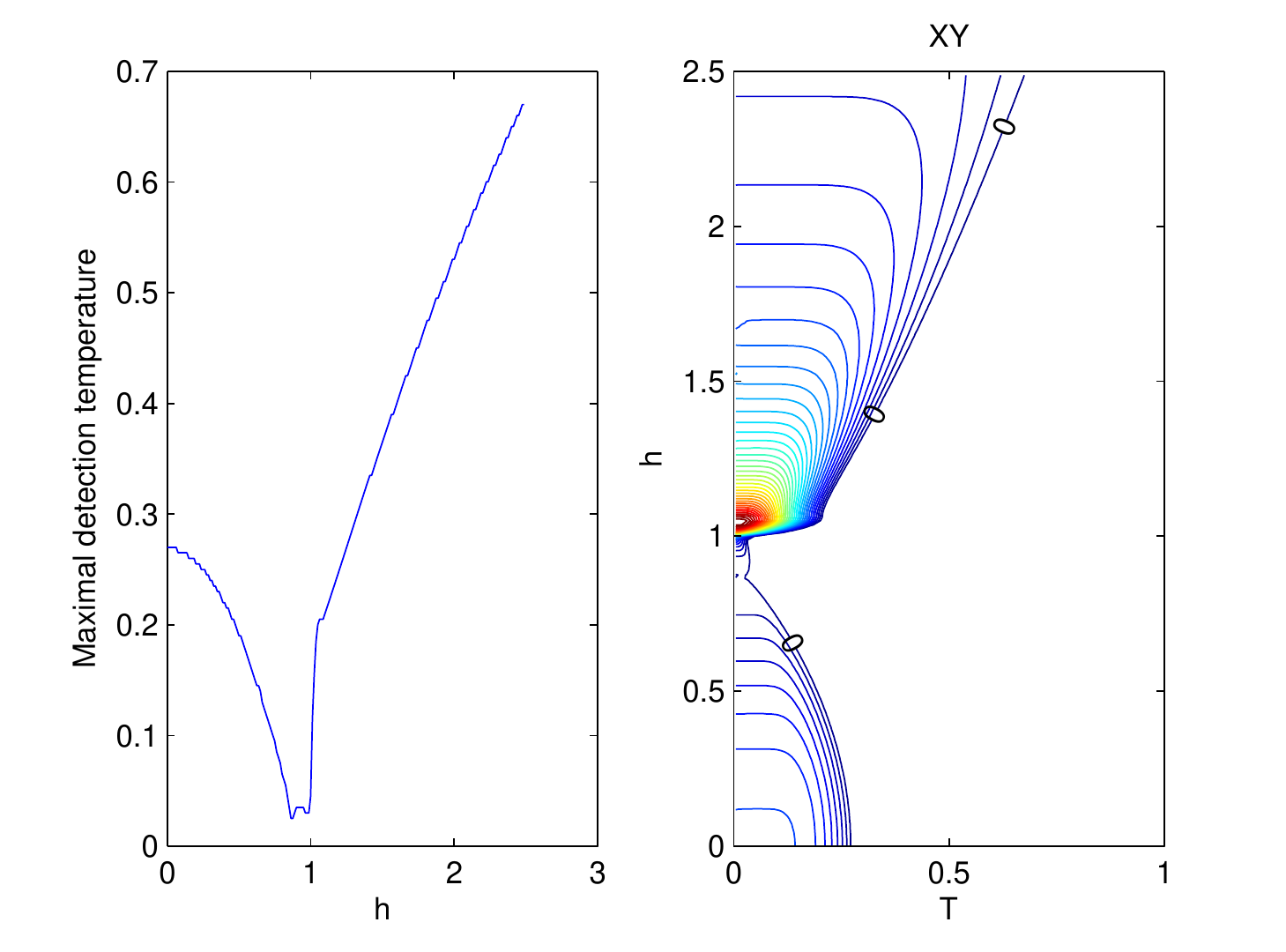}
	(c)\includegraphics[width=0.6\textwidth]{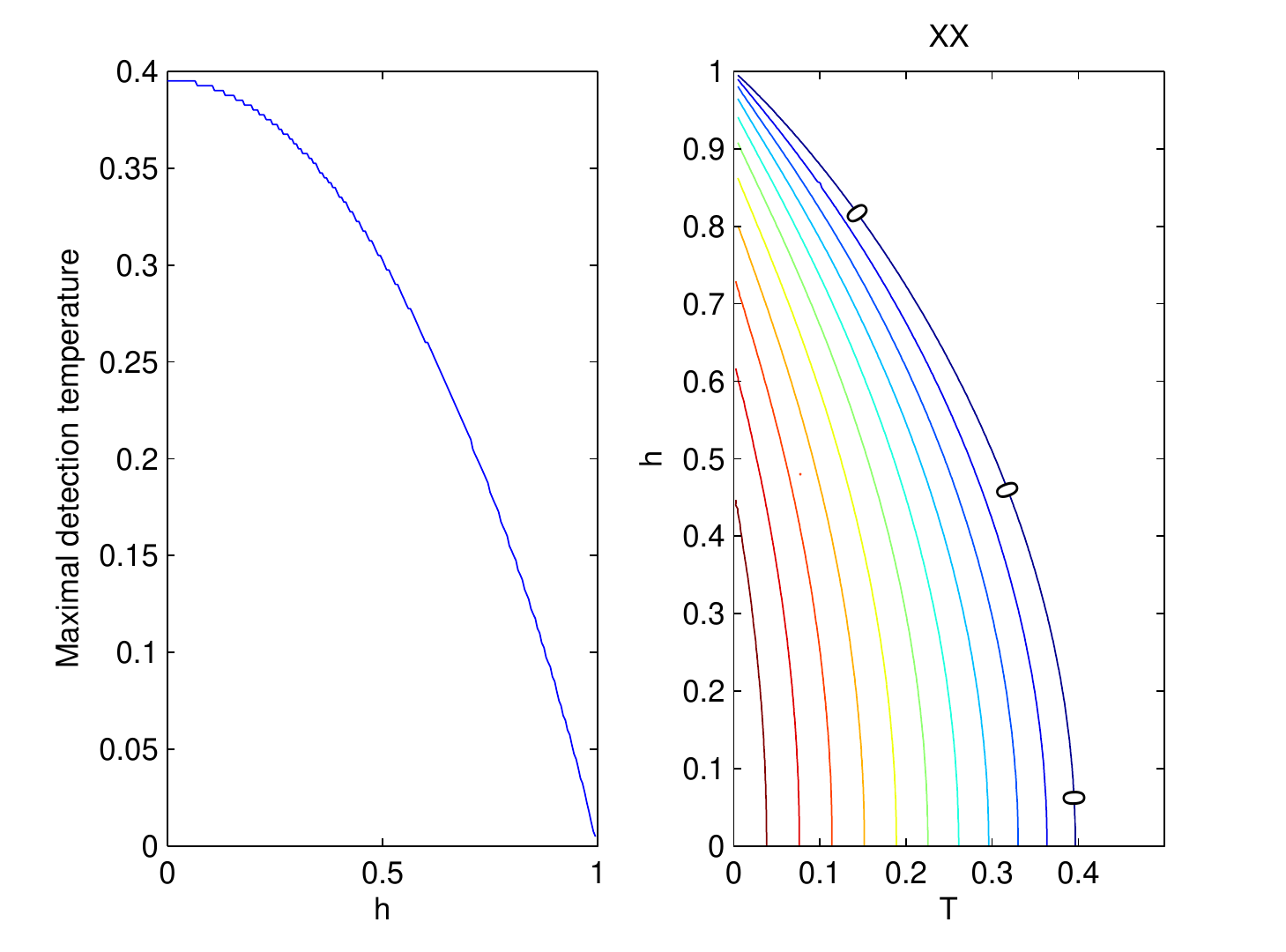}
	\caption{\label{XY_fig}Our results for the (a) Ising, (b) XY with $r=0.5$ and (c) XX models. The l.h.s. show the maximum temperature at which Theorem \ref{alphaResult} yields non-trivial results. The r.h.s is a contour plot of the positive part of $\lim_{N\to\infty}\left(-\frac{\ln\a}{N}-\frac{\beta E_0}{N}-\frac{\ln Z}{N}\right)$. We can show entanglement where this expression is positive, i.e. left of the zero line.}
\end{figure}

 \subsection{Arbitrarily large Gap}
Unfortunately the entropy gap closes in the limit when using the regular ground state witnesses, a fact which is related to the lack of robustness
of the ground state witness itself. It would be most interesting to find other, asymptotically robust witnesses, which could result in 
a macroscopic entropy gap. Let us now present an artificial example of a Hamiltonian which, with the right choice of witness, allows for an arbitrarily big entropy gap in the thermodynamic limit, thus proving the feasibility of this endeavour in principle.

The Hamiltonian is given by
\be
H=\sum_{k=0}^{d^2-1}k\pro{\Psi_k},
\ee
where $d=2^\frac{n}{2}$ and  $\{\ket{\Psi_k}\}_{k=0}^{d^2-1}$ are a Bell state basis of an $n$-qubit Hilbert space. The basis is ordered in such a way that $\ket{\Psi_k}=\frac{1}{\sqrt{d}}\sum_{i=0}^{d-1}\ket{i}\ket{i+k}$ for $k=0...d-1$. Note that these first $d$ basis elements have orthogonal support. As witness, we choose
\be
W=\alpha\1-\sum_{k=0}^{d-1}\pro{\Psi_k},
\ee
which, again, is diagonal in the energy eigenbasis. $\alpha$ is the defined as the maximal overlap of the projector $\sum_{k=0}^{d-1}\pro{\Psi_k}$ with a separable state. As all the $\ket{\Psi_k}$ have orthogonal support, $\alpha$ is the maximal overlap of any of the $\ket{\Psi_k}$ with a separable state, hence $\alpha=\frac{1}{d}$ because the $\ket{\Psi_k}$ are maximally entangled. Inserting this into (\ref{eq:EntropyGapBound}), we obtain
\be
\Delta S\ge\ln\frac{\sum_{k=0}^{d^2-1}e^{-\beta k}}{\sum_{k=0}^{d-1}e^{-\beta k}e^{-\nu(1-\frac{1}{d})}+\sum_{k=d}^{d^2-1}e^{-\beta k}e^{\frac{\nu}{d}}}.
\ee
Letting $n$, hence $d$, go to infinity, we obtain
\be
\lim_{n\to\infty}\Delta S\ge\nu,
\ee
where we have used the fact that $\lim_{d\to\infty}\sum_{k=0}^{d^2-1}e^{-\beta k}=\lim_{d\to\infty}\sum_{k=0}^{d-1}e^{-\beta k}=\frac{e^\beta}{e^\beta-1}$. Note that $\nu$ can be chosen arbitrarily large.
\section{Discussion}
\subsection{Robustness}
In macroscopic systems assuming the exact form of the Hamiltonian is 
always an idealisation. While it is fair from a physical 
point of view one might wonder about the impact of mis-characterised Hamiltonians 
on our entropy gap $\Delta S$. If we assume that the real Hamiltonian is given as $\tilde{H}=H+\epsilon P$ we can that the dual entropy of the actual Hamiltonian for the chosen parameters 
\begin{align}
  \label{eq:max-S-W-dual-P}
  \tilde{S}_{\text{max},W,E}(\mu,\nu) =  \log\tr\exp(-\mu H -\mu\epsilon P+ \nu W) + \mu E+\mu\epsilon \text{Tr}(\rho P) \,,
\end{align}
can be bounded from above using the Golden-Thompson $\text{Tr}(e^{A+B})\leq \text{Tr}(e^{A})\text{Tr}(e^{B})$ and H\"older's $||AB||_1\leq||A||_\infty||B||_1$ inequalities to yield
\begin{align}
  \tilde{S}_{\text{max},W,E}(\mu,\nu) \leq  S_{\text{max},W,E}(\mu,\nu) +\mu\epsilon \text{Tr}(\rho P) -\mu\epsilon P_o\,.
\end{align}
This shows that small perturbations or inaccuracies in the description of the Hamiltonian will only have a correspondingly small impact on the validity of the entanglement certification by \emph{proxy witnesses}.

\subsection{Experimental estimation of entropy}

While the mean energy and the entropy are both macroscopic properties 
of quantum states, measurements of the latter are not possible 
directly as they do not correspond to a quantum observable. There are 
however various ways that the global entropy can be determined, and 
all we need is a lower bound on the entropy. The most straightforward 
way would of course consist of equilibrating the system with a thermal 
bath at temperature $T$. This is naturally the case, as any system found
in nature at ambient temperature $T$, will most accurately be described 
by its corresponding Gibbs state. Thus, once the system is equilibrated we know that 
its corresponding entropy should correspond to the Gibbs entropy 
$S\bigl(\rho_{\text{th}}(\beta(E))\bigr)$. Starting from equilibrated 
systems one can introduce global quenches of system parameters, such 
that the entropy will still be bounded from below by the initial Gibbs 
entropy due to the second law of thermodynamics, despite the system being far out of equilibrium.
In that way one directly receives a lower bound on the system entropy for a wide range 
of out of equilibrium systems and experimental preparations.
There are of course other methods, such as reasonable assumptions about symmetry in the state, that can be used 
to estimate linear entropy from macroscopic spin observables (as would be possible e.g. in Ref.\cite{linent}).
Another straightforward way to obtain lower bounds would be access to a $d\times d$ sub-matrix of the global state,
whose entropy will always yield a lower bound to the global entropy. The exact ways of experimentally estimating entropies
is of course highly dependent on the experimental setup, access and reasonable assumptions about system properties. As any
non-zero amount of entropy already provides an advantage in entanglement detection it is however fairly straightforward to
infer correspondingly useful entropies.

This highlights the main advantage of our approach: While entanglement 
of Gibbs states could be directly inferred from its description, 
it would strictly work only if we know that the state is in fact in
thermal equilibrium. 
This would require strict assumptions about the state of the system and exact 
characterization of its Hamiltonian. Using our entropy gap there is 
no need whatsoever to assume any particular form of the system's state: 
measurements of the mean energy and lower bounds on the global 
entropy are completely sufficient to prove that the underlying system 
is entangled in a robust way. We believe that future work will
uncover witnesses robust enough to be amenable to the proxy method,
yielding a sizeable entropy gap in the thermodynamic limit.
\section{Conclusion}
We have introduced a framework in which concave functions can be used as a \emph{proxy} for detecting entanglement in many body systems. As the most relevant example we have explored the use of entropy in this context, yielding \emph{separable Gibbs ensembles} and a corresponding \emph{entropy gap} that can be used for entanglement detection. Through the efficient description through \emph{witness canonical ensembles} this unlocks powerful tools from entanglement theory, such as positive maps or entanglement witnesses to be harnessed in situations where the actual estimation of said quantities is experimentally impossible.

We hope that our work can contribute to the further understanding of the relationship between important physical aspects of many-body systems, such as phase transitions, and the paradigmatic feature of Quantum Information Theory, entanglement \cite{Amico2008}. We address some possible new avenues in entanglement detection, by indicating the importance of an open conjecture regarding the geometric measure of entanglement for ground states \cite{wei2005global} and by specifying particularly useful forms of entanglement witnesses in many-body systems.

We have furthermore demonstrated the usefulness and feasibility of this approach in exemplary and paradigmatic physical Hamiltonians. For moderate system size we have demonstrated how our results directly improve upon previous work on inferring entanglement from macroscopic observables. The resulting entropy gaps quantify to what extent even out-of-equilibrium systems can be certified to exhibit entanglement. In the thermodynamic limit we have shown that particularly promising paths towards this goal are strongly connected to ground state properties of many-body Hamiltonians. Our results in this context furthermore elucidate how a development of robust entanglement witness techniques can prove useful, even if they themselves remain inaccessible due to experimental limitations.

\bigskip
\textbf{Acknowledgements}
The authors are pleased to acknowledge enlightening discussions on
many-body systems and quantum thermodynamics with Janet Anders, 
Mart\'{\i} Perarnau and Anna Sanpera. 
SB would like to thank Johann L\"ofberg for his 
help including nonlinear objectives into SDPs.
SB and AW are supported by the European Research Council (Advanced Grant ``IRQUAT''). 
AW and MH were supported by the European Commission (STREP ``RAQUEL''),  
by the Spanish MINECO, projects FIS2008-01236 and FIS2013-40627-P, 
with the support of FEDER funds, and by the Generalitat de Catalunya CIRIT,
project 2014-SGR-966. 
MH furthermore acknowledges funding from the Juan de la Cierva fellowship (JCI 2012-14155). 
DB and HK acknowledge support from Deutsche Forschungsgemeinschaft (DFG) and Bundesministerium f\"ur
Bildung und Forschung (BMBF).

\bibliographystyle{unsrtnat}

\bibliography{proxy}

\end{document}